\documentclass[10pt,journal]{IEEEtran}

\usepackage{mathrsfs}
\usepackage{cite} 
\usepackage{url}
\usepackage{epsfig}
\usepackage{amsmath}
\usepackage{amssymb}
\usepackage{amsfonts}
\usepackage{graphicx}
\usepackage{subfigure}
\usepackage{wrapfig}
\usepackage{float}
\usepackage{multicol}
\usepackage{balance}
\usepackage{framed}

\newtheorem{theorem}{Theorem}
\newtheorem{lemma}{Lemma}

\newtheorem{definition}{Definition}

\newtheorem{conjecture}{Conjecture}

\newcommand{\st}{\textsf{s}}
\newcommand{\ct}{\textsf{c}}
\newcommand{\pb}{\textsf{p}}
\newcommand{\ts}{\textbf{\emph{t}}}

\usepackage{stfloats}

\begin{document}

\title{Relay Selection with Channel Probing in Sleep-Wake Cycling 
Wireless Sensor Networks}
\author{K.~P.~Naveen,~\IEEEmembership{Student~Member,~IEEE,} and 
Anurag~Kumar,~\IEEEmembership{Fellow,~IEEE} 
\thanks{Both the authors are with the Department of Electrical Communication
Engineering, Indian Institute of Science, Bangalore 560012, India.
Email: \{naveenkp, anurag\}@ece.iisc.ernet.in}
\thanks{This work was supported in part by the Indo-French Centre for the Promotion of 
Advanced Research (IFCPAR Project 4000-IT-1), and in part by the Department of 
Science and Technology (DST) via a J.C.\ Bose Fellowship.}}
\maketitle

\begin{abstract}
In geographical forwarding of packets in a large wireless sensor network (WSN) with sleep-wake 
cycling nodes, we are interested in the local decision problem faced by a node that has 
``custody'' of a packet and has to choose one among a set of next-hop relay nodes to forward the 
packet towards the sink. Each relay is associated with a ``reward'' that 
summarizes the benefit of forwarding the packet through that relay. We seek a  
solution to this local problem, the idea being that such a solution, if 
adopted by every node, could provide a reasonable heuristic for the end-to-end forwarding problem.
Towards this end, we propose a \emph{relay selection problem} comprising a forwarding node
and a collection of relay nodes, with the relays waking up sequentially at random times.
At each relay wake-up instant the forwarder can choose to \emph{probe} a relay to learn its reward 
value, based on which the forwarder can then decide whether to \emph{stop} (and forward its packet 
to the chosen relay) or to \emph{continue} to wait for further relays to wake-up. The forwarder's 
objective is to select a relay so as to minimize a combination of waiting-delay, reward and 
probing cost. Our problem can be considered as a variant of the asset selling problem studied 
in the operations research literature. We formulate our relay selection problem as a
Markov decision process (MDP) and obtain some interesting structural results on the 
optimal policy (namely, the threshold and the stage-independence properties). 
We also conduct simulation experiments and gain valuable insights into the performance of our 
local forwarding-solution. 
\end{abstract}


\section{Introduction}
Consider a wireless sensor network deployed for the detection of \emph{rare events},
e.g., forest fires, intrusion in border areas, etc. To conserve energy, the nodes in the network 
\emph{sleep-wake} cycle whereby they alternate between an ON state and a low power OFF state.
We are further interested in \emph{asynchronous} sleep-wake cycling where the point processes
of wake-up instants of the nodes are not synchronized
\cite{kim-etal09optimal-anycast,naveen-kumar12relay-selection_TMC_paper}.

In such networks, whenever an event is detected, an alarm packet (containing the event location
and a time stamp) is generated and has to be forwarded, through multiple hops 
(as illustrated in Fig.~\ref{global_figure}), to a control center (\emph{sink}) where appropriate 
action could be taken. Since the network is sleep-wake cycling, a forwarding node (i.e., a node 
holding an alarm packet) has to wait for its neighbors to wake-up before it can choose one for 
the next hop. Thus, due to the sleep-wake process, there is a delay incurred at each hop en-route 
to the sink, and our interest is in minimizing the total average end-to-end delay subject to a 
constraint on some global metric of interest such as the average hop count, or the average total 
transmission power (sum of the transmission power used at each hop). Such a  global problem can be 
considered as a stochastic shortest path problem 
\cite{bertsekas-tsitsiklis91stochastic-shortest-path}, 
for which  the distributed Bellman-Ford algorithm (e.g., the LOCAL-OPT algorithm 
proposed by Kim et al.\ in \cite{kim-etal09optimal-anycast}) can be used to obtain the optimal 
solution.  However, a major drawback with such an approach is that a pre-configuration phase is required to 
run such algorithms, which would involve exchange of several control messages.
Furthermore, such global configuration would need to be performed each time there is a change in 
the network topology, such as due to node failures, or variations in the propagation characteristics, etc.

\begin{figure}
\centering
\includegraphics[scale=0.55]{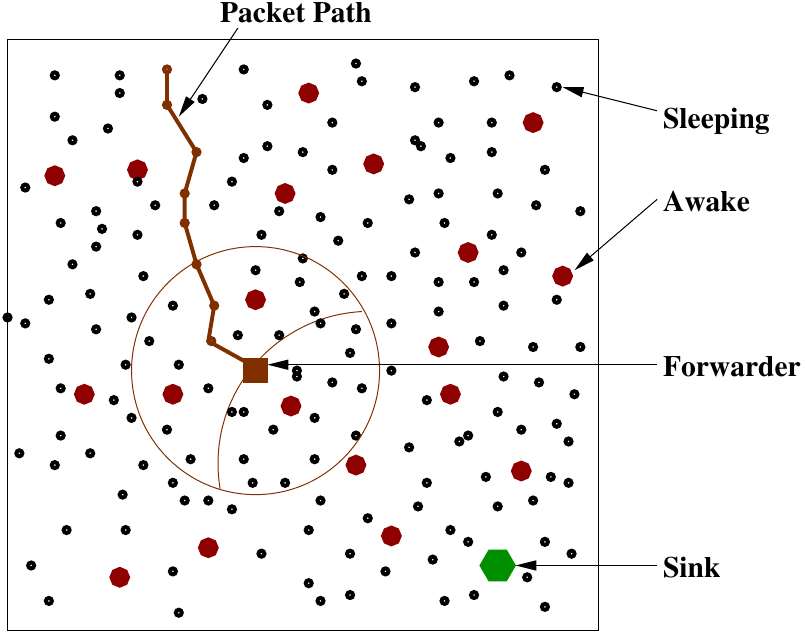}
\caption{\label{global_figure} Illustration of a packet being forwarded to the sink node (green
hexagon) through a sleep-wake cycling network. The square node (labeled as forwarder) is the 
current custodian of the packet.}
\vspace{-6mm}
\end{figure}

The focus of our research is instead towards designing \emph{simple forwarding rules} 
that use only the \emph{local information} available at a forwarding node.
In our own earlier work in this direction 
\cite{naveen-kumar10geographical-forwarding,naveen-kumar12relay-selection_TMC_paper}, we 
formulated the local forwarding problem as one of minimizing the one-hop forwarding delay subject 
to a constraint on the reward offered by the chosen relay. The reward associated with a relay is a 
function of the transmission power and the progress towards the sink made by the packet when 
forwarded via that relay. We considered two variations of the problem, one in which the number of 
potential relays is known \cite{naveen-kumar10geographical-forwarding}, and the other in which 
only a probability mass function of the number of potential relays is known 
\cite{naveen-kumar12relay-selection_TMC_paper}. In each case, we derived the structure of the 
optimal policy. Further, through simulation experiments we found that, in some 
region of operation, 
the end-to-end performance (i.e., total delay and total transmission power)
obtained by applying the solution to the local problem at each hop is comparable with that 
obtained by the global solution (i.e., the LOCAL-OPT proposed by Kim et al.\ 
\cite{kim-etal09optimal-anycast}), 
thus providing additional support for the approach of utilizing local 
forwarding rules, albeit suboptimal.

In our earlier work, however, we assume that the gain 
of the wireless communication channel between
the forwarding node and a relay is a deterministic function of the distance between
the two, whereas, in practice, due to the phenomenon called \emph{shadowing}, the channel 
gain at a given distance from the forwarding node is not a constant, but varies spatially over 
points at the same distance (the variation being typically modeled as log normally 
distributed \cite{rappaport01wireless-communication}). In addition to not being just a function of 
distance,  the path-loss between a pair of locations varies with time; in a forest, for example, 
this would be due to seasonal variations in the foliage.
Therefore, in each instance that a node gets custody of a packet, the node has to send probe 
packets to determine the channel gain to relay nodes that wake up, and thereby ``offer'' 
to forward the packet. Such probing incurs additional cost (for instance, see 
\cite{thejaswi-etal10two-level-probing} where probing allows the transmitter to obtain a finer estimate of the channel gain). 
Hence, ``to probe'' or ``not to probe'' can itself 
become a part of the decision process.
In the current work we incorporate these features (namely, channel probing and the 
associated power cost) while choosing 
a relay for the next hop, leading to an interesting variant of the asset selling problem
\cite[Section~4.4]{bertsekas05optimal-control-vol1},
\cite{karlin62selling-asset} studied in the operations research literature.

\textbf{\emph{Outline and Our Contributions:}} In Section~\ref{local_forwarding_problem_section}
we will formally describe our system model,
following which we will discuss the related work. Sections~\ref{restricted_class_section} 
and \ref{structural_results_section} are devoted towards characterizing the structure 
of the policy RST-OPT (ReSTricted-OPTimal) which is optimal within a restricted class of relay 
selection policies. In Section~\ref{general_class_section} we will discuss the globally optimal 
GLB-OPT policy. Numerical and simulation results are presented in 
Section~\ref{probing_numerical_work_section}. 
Our main technical contributions are the following:
\begin{itemize}
\item We characterize the optimal policy, RST-OPT, in terms of \emph{stopping} sets. 
We prove that the stopping sets have a threshold structure (Theorem~\ref{threshold_nature_lemma}).

\item We further prove that the stopping sets are identical across the decision stages 
(Theorem~\ref{stopping_sets_equal_theorem} and \ref{stopping_sets_l_equal_theorem}).
This result can be considered as a generalization of the \emph{one-step-look ahead}  rule 
(see the remark following Theorem~\ref{stopping_sets_equal_theorem}).

\item Through one-hop numerical work we find that the performance of 
RST-OPT is close to that of GLB-OPT. This result is useful because, the sub-optimal RST-OPT 
is computationally more simpler than GLB-OPT. We have also conducted 
simulations to study the end-to-end performance of RST-OPT.
\end{itemize}
We will finally conclude in Section~\ref{prb_conclusion_section}.
For the ease of readability we have moved most of the proofs to the Appendix.

\section{System Model: The Relay Selection Problem}
\label{local_forwarding_problem_section}
We will describe the system model from the context of \emph{geographical forwarding}, also known 
as location aware routing,
\cite{naveen-kumar10geographical-forwarding,akkaya-younis05survey,zorzi-rao03geographicrandom}.
In geographical forwarding it is assumed that each node in the network knows its location (with 
respect to some reference) as well as the location of the sink.

\begin{figure}[t]
\centering
\includegraphics[scale=0.6]{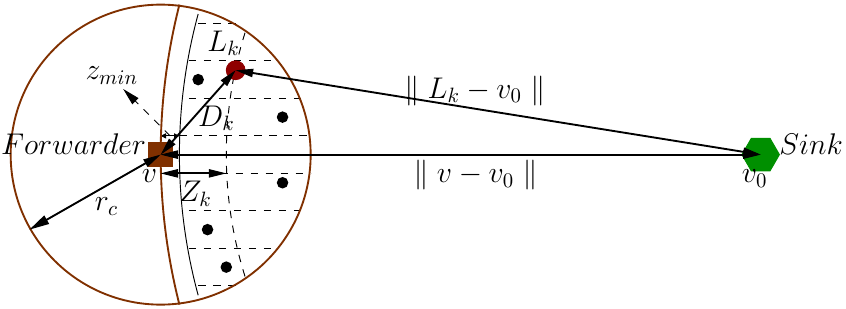} 
\caption{\label{forwarding_set_figure} The hatched area is the forwarding region $\mathcal{L}$. 
For $\ell\in\mathcal{L}$, the progress $Z_\ell$ is the difference between the forwarder-to-sink 
and $\ell$-to-sink distances.}
\vspace{-6mm}
\end{figure}

Consider a forwarding node $\mathscr{F}$ located at $v$ (see Fig.~\ref{forwarding_set_figure}). 
The sink node is situated at $v_0$. Thus, the distance between $\mathscr{F}$ and the sink is 
$V=\ \parallel v-v_0\parallel$ (we use $\parallel\cdot\parallel$ to denote 
the Euclidean norm). The \emph{communication region} is the set of all locations 
where reliable exchange of \emph{control messages} (transmitted using a
low rate robust modulation technique on a separate control channel)
can take place between $\mathscr{F}$ and a receiver, if any, at these locations.  
In Fig.~\ref{forwarding_set_figure}  we have shown the communication region to be circular, 
but in practice this region can be arbitrary. The set of nodes within the communication region 
are referred to as the \emph{neighbors}. 

Let $V_{\ell}=\ \parallel\ell-v_0\parallel$ represent the 
distance of a location $\ell$ (which is a point in $\Re^2$) from the sink. Now define
the \emph{progress} of location $\ell$ as $Z_{\ell}=V-V_{\ell}$, which is simply the difference 
between the $\mathscr{F}$-to-sink and $\ell$-to-sink distances. $\mathscr{F}$ is
interested in forwarding the packet only to a neighbor within the \emph{forwarding region} 
$\mathcal{L}$, which is defined as
\begin{eqnarray}
\label{prb_forwarding_region_figure}
\mathcal{L}=\Big\{\ell\in\mbox{communication region}: Z_{\ell}\ge z_{min}\Big\}
\end{eqnarray}
where, $z_{min}>0$ is the minimum progress constraint (see Fig.~\ref{forwarding_set_figure},
where the hatched area is the forwarding region).
The reason for using $z_{min}>0$ in the definition of $\mathcal{L}$ are:
(1) practically this will ensure that a progress of at least $z_{min}$ is made 
by the packet at each hop, and (2) mathematically this condition will allow us to bound the 
reward functions (to be defined sooner) to take values within an interval 
$[0,\overline{r}]$. 

Next, it is natural to assume that $\mathcal{L}$ is bounded; we will further 
assume that $\mathcal{L}$ is closed. The reason for imposing this condition will become clear in 
Section~\ref{structural_results_section}. Finally, we will refer to the nodes in the forwarding 
region as \emph{relays}.

\textbf{\emph{Sleep-Wake Process:}} 
Without loss of generality, we will assume that 
$\mathscr{F}$ receives an alarm packet (from an upstream node) at time $0$, which has to be
forwarded to one of the relays. There are $N$ relays that wake-up sequentially at the points of 
a Poisson process of rate $\frac{1}{\tau}$.\footnote{\label{sleep-wake-footnote}A practical approach
for sleep-wake cycling is the \emph{asynchronous periodic} process, where each relay $i$ 
wakes up at the periodic instants $T_i + kT$ with $\{T_i\}$ being i.i.d.\ (independent and 
identically distributed) uniform on $[0,T]$ 
\cite{kim-etal09optimal-anycast,naveen-kumar12relay-selection_TMC_paper}.
Now, for large $N$ if $T$ scales with $N$ such that $\frac{N}{T}\rightarrow\frac{1}{\tau}$,  
then the aggregate point process of relay wake-up instants converges to a Poisson process of rate 
$\frac{1}{\tau}$ \cite{cinlar75stochastic-processes}, thus justifying our Poisson process 
assumption.} The wake-up times are denoted, $0\le W_1\le \cdots \le W_N$. The relay waking 
up at the instant $W_k$ is referred to as the $k$-th relay. Let  $U_1=W_1$ and $U_k=W_k-W_{k-1}$ 
$(k=2,\cdots,N$) denote the \emph{inter-wake-up time} between the $k$-th and the $(k-1)$-th relay. 
Then, $\{U_k\}$ are i.i.d.\ exponential random variables with mean $\tau$.

\textbf{\emph{Channel Model:}}
We will consider the following standard model for the transmission power required 
by $\mathscr{F}$ to achieve an 
SNR (signal to noise ratio) constraint of $\Gamma$ at some location $\ell$, 
whose distance from $\mathscr{F}$ is more than $d_{ref}$  
(far-field reference distance beyond which the following expression will hold \cite{kumar-etal08wireless-networking}):
\begin{eqnarray}
\label{power_equn}
P_{\ell}=\frac{\Gamma N_0}{G_{\ell}}{\left(\frac{D_{\ell}}{d_{ref}}\right)}^{\xi}
\end{eqnarray}
where,
$D_{\ell}=\ \parallel \ell-v\parallel$ is the distance between $\mathscr{F}$
and $\ell$, $G_\ell$ is the random component of the channel 
gain between $\mathscr{F}$ and $\ell$,
$N_0$ is the receiver noise variance, and
$\xi$ is the path-loss attenuation factor.
We will assume that $d_{ref}\le z_{min}$ so that $P_\ell$ in (\ref{power_equn}) is the
power required for any $\ell\in\mathcal{L}$. Also, for simplicity, from here on we
will use $\Gamma'$ to  denote $\frac{\Gamma N_0}{d_{ref}^\xi}$.

Although $G_\ell$ along with the path-loss, ${\left(\frac{D_{\ell}}{d_{ref}}\right)}^{\xi}$,
constitutes the gain of the channel, for simplicity we will throughout refer to $G_\ell$ 
itself as the channel gain between $\mathscr{F}$ and the location $\ell$. 
We will assume that the set of channel gains, 
$\{G_\ell:\ell\in\mathcal{L}\}$, are i.i.d. We will further 
assume that the channel coherence time is large 
so that the channels gains remain unchanged over the entire duration of the decision process, i.e.,
in physical layer wireless terminology, we have a \emph{slowly varying channel}.

\emph{Remark:} There are two remarks we would like to make here. First is regarding
the channel gains being i.i.d.
Since the randomness in the channel is spatially correlated 
\cite{agrawal-patwari09correlated-shadow_journal}, 
if two locations $\ell$ and $u$ are very 
close then the corresponding gains, $G_\ell$ and $G_u$, will not be independent; a minimum
separation between the receivers is required for the gains to be statistically independently.
Thus, our assumption of independence between the channel gains to the relays
requires that the relays should not be close to each other, or, equivalently, the relay density 
should not be large. We will assume that this physical property holds, and, thus, proceed with the 
technical assumption that the channel gains are i.i.d. 

Next, about the slowly varying channel, 
it suffices 
for the channel coherence time to be longer than the sleep-wake cycling period
(recall footnote~\ref{sleep-wake-footnote} from page~\pageref{sleep-wake-footnote}). Under our light 
traffic assumption where the events are rare, with a probability close to $1$, a 
node wakes up and finds no forwarding node 
in its communication range. Thus, with a high probability, when a node wakes up, it stays awake for a 
few milliseconds, e.g., $3$ milliseconds (for sending a control packet, turning the radio from 
send to listen, and then waiting for a possible response).
Thus, for example, with a $1\%$ duty cycle, the inter-wakeup 
time would need to be $300$ milliseconds, imposing a reasonable requirement 
on the channel coherence time.



\textbf{\emph{Reward Structure:}} 
Finally, combining progress, $Z_\ell$, and power, $P_\ell$, we define the reward 
associated with a location $\ell\in\mathcal{L}$ as,  
\begin{eqnarray}
\label{reward_equn}
R_{\ell}=\frac{Z_{\ell}^a}{P_{\ell}^{(1-a)}}
=\frac{Z_{\ell}^a}{(\Gamma' D_{\ell}^\xi)^{(1-a)}} G_{\ell}^{(1-a)},
\end{eqnarray}
where $a\in[0,1]$ is used to trade-off between $Z_{\ell}$ and $P_{\ell}$. 
The reward being inversely proportional to $P_{\ell}$ is clear because it is advantageous
to use low power to get the packet across; $R_{\ell}$ is proportional to $Z_{\ell}$
to promote progress towards the sink while choosing a relay for the next hop. 

The channel gains, $\{G_\ell\}$, are non-negative; we will further assume that they are bounded 
above by $g_{max}$. These conditions along with $Z_\ell\ge z_{min}$ (which implies that 
$D_\ell\ge z_{min}$) and $\mathcal{L}$ is bounded (so that $Z_\ell\le z_{max}$ for all 
$\ell\in\mathcal{L}$) will provide the following upper bound for the reward functions 
$\Big\{R_\ell:\ell\in\mathcal{L}\Big\}$:
\begin{eqnarray*}
\overline{r}
&=&  \frac{z_{max}^a}{(\Gamma' z_{min}^\xi)^{(1-a)}} g_{max}^{(1-a)}.
\end{eqnarray*}
Thus, the reward values lie within the interval $[0,\overline{r}]$.
 
Let $F_{\ell}$ represent the c.d.f.\ (cumulative distribution function) of $R_{\ell}$, and
\begin{eqnarray}
\label{distribution_set}
\mathcal{F}
&=& \Big\{F_\ell:\ell\in\mathcal{L}\Big\}
\end{eqnarray}
denote the collection of all possible
reward distributions.
From (\ref{reward_equn}), note that, given a location $\ell$ it is only possible
to know the reward distribution $F_\ell$. To know the exact reward $R_\ell$,
$\mathscr{F}$ has to \emph{transmit probe packets} to learn the channel gain $G_\ell$
(we will formalize probing very soon).

\textbf{\emph{Relay Locations:}}
We will assume that  each of the $N$ relays is randomly and mutually independently located in the 
forwarding region $\mathcal{L}$. Formally, if  $L_1,L_2,\cdots,L_N$ denotes the relay locations, 
then these are i.i.d.\ uniform over the forwarding set $\mathcal{L}$ (this assumption holds if the 
nodes are deployed according to a spatial Poisson process). Let $L$ denote the uniform distribution 
over $\mathcal{L}$ so that, for $k=1,2,\cdots,N$, the distribution of $L_k$ is $L$. 

\emph{Remark:}
Although, for 
the sake of motivating the model, we have restricted to a very specific $\mathcal{F}$ (set of 
reward distributions) and $L$ (relay location distribution), it is important to note that
all our analysis in the subsequent sections will follow through for more general $\mathcal{F}$
and $L$ as well. 

At time $0$, $\mathscr{F}$ only knows that there are $N$ relays in its 
forwarding set $\mathcal{L}$, but  does not  know their locations, $L_k$, nor their 
channel gains, $G_{L_k}$. 

\textbf{\emph{Sequential Decision Problem:}}
When the $k$-th relay wakes up, we assume that its location $L_k$, and hence its reward 
distribution $F_{L_k}$ is revealed to $\mathscr{F}$. This can be accomplished by including the 
location information $L_k$ within a control packet (sent using a low rate robust modulation 
technique, and hence, assumed to be error free) transmitted by the $k$-th relay upon waking up. 
However, if $\mathscr{F}$ wishes to learn the channel gain $G_{L_k}$ (and hence the exact reward 
value $R_{L_k}$), it has to transmit  additional probe packets (indeed several packets) in order to 
obtain a reliable estimate of the channel gain, incurring a power cost of $\delta\ge0$ units. 
Thus, when the $k$-th relay wakes up (referred to as \emph{stage} $k$), given the set of previously
probed and unprobed relays (i.e., the history), the actions available to $\mathscr{F}$ are:
\begin{itemize}
\item $\st$: \emph{stop} and forward the packet to a relay with the maximum reward 
(\emph{best relay}) among the probed relays; with this action the decision process ends.

\item $\ct$: \emph{continue} to wait for the next relay to wake-up (average waiting time 
is $\tau$); with this action the decision process enters stage $k+1$. 

\item $\pb$: {probe} a relay from the set of all unprobed relays (provided there is at least 
one unprobed relay). The probed relay's reward value is then revealed, allowing $\mathscr{F}$
to update the best relay. \emph{After probing, the decision process is still at stage $k$ and 
$\mathscr{F}$ has to again decide upon an action}.
\end{itemize}

In the model, for the sake of analysis, we neglect the time taken for the exchange of control 
packets and the time taken to probe a relay to learn its channel gain. We argue that this is 
reasonable for very low duty cycling networks, where the average inter-wake-up time is much larger 
than the time taken for probing and for the exchange of control packets.

At stage $k$, let $b_k$ denote the reward of the best relay, and $\mathcal{F}_k$ be the vector of 
reward distribution of the unprobed relays, i.e., formally, 
\begin{eqnarray*}
b_k=\max\Big\{R_{L_i}: i\le k, \mbox{ relay } i \mbox{ has been probed}\Big\},
\end{eqnarray*}
and
\begin{eqnarray*}
\mathcal{F}_k=\Big(F_{L_i}: i\le k, \mbox{ relay } i \mbox{ is unprobed}\Big).
\end{eqnarray*}
We will regard $(b_k,\mathcal{F}_k)$ to be the state of the system at stage $k$.
Note that, it is possible that until stage $k$ no relay has been probed, in which 
case $b_k=-\infty$, or all the relays are probed so that $\mathcal{F}_k$ is empty.
Whenever $\mathcal{F}_k$ is empty we will represent the state as simply $b_k$.
Now we can define a forwarding policy $\pi$ as follows:
\begin{definition}
A policy $\pi$ is a sequence of mappings $(\mu_1,\mu_2,\cdots,\mu_N)$ where, 
\begin{itemize}
\item for $k=1,2,\cdots,N-1$, 
$\mu_k(b_k,\mathcal{F}_k)\in\{\st,\ct,\pb\}$ and $\mu_k(b_k)\in\{\st,\ct\}$, and

\item $\mu_N(b_N,\mathcal{F}_N)\in\{\st,\pb\}$ and $\mu_N(b_N)\in\st$.
\end{itemize}
Note that the action to continue is not available at the last stage $N$. Let $\Pi$
denote the set of all policies. 
\hfill $\blacksquare$
\end{definition}

For a policy $\pi\in\Pi$, the delay incurred, denoted $D$,  is the time until a relay is chosen.
Let $R$ denote the reward offered by the chosen 
relay. Further, let $M$ denote the total number of relays that 
were probed during the decision process.
Then, recalling that $\delta$ is the probing cost, $\delta M$ represents the total 
cost of probing. We would like to think of $(R-\delta M)$ as the 
\emph{effective reward} achieved using policy $\pi$.
Then, denoting $\mathbb{E}[\cdot]$ to be the expectation operator conditioned on 
using policy $\pi$, the problem we are interested in is the following:
\begin{eqnarray}
\label{unconstrained_equn}
\mbox{Minimize}_{\pi\in\Pi}\ \bigg(\mathbb{E}_\pi[D] - \eta \Big(\mathbb{E}_\pi[R] -
\delta\mathbb{E}_\pi[M]\Big)\bigg),
\end{eqnarray}
where $\eta>0$ is the multiplier used to trade-off between delay and effective reward.

\textbf{\emph{Restricted Class $\overline{\Pi}$:}}
Recall that the state at stage $k$ is of the form $(b_k,\mathcal{F}_k)$
where $\mathcal{F}_k$ is the set of all unprobed relays. The size of $\mathcal{F}_k$
can vary from $0$ (if all the $k$ relays that have woken up thus far have been probed) to $k$ (if 
none have been probed). Further, suppose the size of $\mathcal{F}_k$ is $m$ ($0<m\le k$) then 
$\mathcal{F}_k\in\mathcal{F}^m$ (the $m$ times Cartesian product of $\mathcal{F}$)
since the reward distribution of
each unprobed relay can be any distribution from $\mathcal{F}$.
Thus, the set of all possible states at stage $k$ is large.
Hence, for analytical tractability, we first  consider (in Sections~\ref{restricted_class_section} 
and \ref{structural_results_section})
solving the problem in (\ref{unconstrained_equn}) over a \emph{restricted class} of policies,
$\overline{\Pi}\subseteq\Pi$, where a policy is restricted to take decisions keeping only
up to two relays awake $-$ one the best among all probed relays and other the best among 
the unprobed ones. Thus, the decision at stage $k$
is based on $(b_k,H_k)$ where $H_k$ is the ``best distribution in $\mathcal{F}_k$''
(our notion of best distribution is based on stochastic ordering; 
we will formally discuss  this  in 
Section~\ref{structural_results_section}). Later in Section~\ref{general_class_section} we will 
discuss the optimal policy within the \emph{unrestricted class} of policies $\Pi$.

\emph{\textbf{Related Work:}}
Suppose the probing cost $\delta=0$, then the objective in (\ref{unconstrained_equn}) will reduce to 
minimizing $\Big(\mathbb{E}_\pi[D] - \eta \mathbb{E}_\pi[R]\Big)$. Further, when $\delta=0$, since 
there is no advantage in not probing, an optimal policy is to always probe relays as they wake-up 
so that their reward value is immediately revealed to $\mathscr{F}$. Alternatively,
if  $\mathscr{F}$ is not allowed to exercise the option to not-probe a relay, then again the model reduces to the 
case where the relay rewards are immediately revealed as and when they wake-up.

We have 
studied this particular case of our relay selection problem (which we will refer to as the 
\emph{basic relay selection model}) in our earlier work 
\cite[Section~6]{naveen-kumar12relay-selection_TMC_paper},\cite{naveen-kumar10geographical-forwarding}, 
and this basic model can be shown to be 
equivalent to a basic version of the \emph{asset selling problem} 
\cite[Section~4.4]{bertsekas05optimal-control-vol1},
\cite{karlin62selling-asset} studied in the operations research literature. 
The asset selling problem comprises a seller (with some asset to sell)
and a collection of buyers who are arriving sequentially in time. The offers made by the
buyers are i.i.d. If the seller wishes to choose an early offer, then he can invest the funds
received for a longer time period. On the other hand, waiting could yield a better
offer, but with the loss of time during which the sale-proceeds could have been invested. 
The seller's objective is to choose an offer so as
to maximize his final revenue (received at the end of the investment period).
Thinking of the offer of a buyer as analogous to the reward of a relay, the seller's objective of 
maximizing revenue is equivalent to the forwarder's objective
of minimizing a combination of delay and reward. 

However, in the present work we generalize this basic version by allowing the probing cost to be
positive (i.e., $\delta>0$) so that a relay's reward value (equivalently, buyer's offer value) is 
not revealed to the forwarder (equivalently, seller) for free. Instead the forwarder can choose to
probe a relay to know its reward value after incurring an additional cost of $\delta$.
Although there is work reported in the asset selling problem literature which is centered around
the idea of the offer (or reward) distribution being unknown, or not knowing a parameter of the 
offer distribution \cite{albright77generalized-house-selling,rosenfield-etal83selling-asset} but
these do not incorporate an additional probe action like in our model here. 
To the best of our knowledge, the particular class of models we study here is not available in the 
asset selling problem literature. 

Problem of choosing a next-hop relay arises in the context of \emph{geographical forwarding}
(as mentioned earlier, \emph{geographical forwarding} 
\cite{akkaya-younis05survey,mauve-hartenstein01survey} 
is a forwarding technique where the prerequisite is that the nodes know their respective locations 
as well as the sink's location). For instance, Zorzi and Rao in \cite{zorzi-rao03geographicrandom} 
propose an algorithm called GeRaF (Geographical Random Forwarding) which, at each forwarding stage,
chooses the relay making the largest progress. For a sleep-wake cycling network, Liu et al.\ in 
\cite{liu-etal07CMAC} propose a relay selection approach as a part of CMAC, a protocol for 
geographical packet forwarding. Under CMAC, node $i$ chooses an $r_0$ that minimizes the
expected normalized latency (which is the average ratio of one-hop delay and progress).
Links to more literature on similar work from the context of geographical forwarding
can be found in \cite{naveen-kumar12relay-selection_TMC_paper}. However, these work do not 
incorporate the action of ``probing a relay'' as in our relay selection model here.

From the context of wireless communication, the action to {probe} generally occurs in the 
problem of channel selection \cite{chaporkar-proutiere08joint-probing,chang-liu07channel-probing}.
For instance, the authors in \cite{chaporkar-proutiere08joint-probing} study the following problem: 
a transmitter, aiming to maximize its throughput, has to choose a channel for its transmissions, 
among several available ones. The transmitter, only knowing the channel gain distributions, has to 
send probe packets to learn the exact channel state information (CSI). Probing many channels 
yields a channel with a good gain but reduces the effective time for transmission within the channel 
coherence period. The problem is to obtain optimal strategies to decide when to stop probing and to 
transmit. An important difference with our work is that, in 
\cite{chaporkar-proutiere08joint-probing,chang-liu07channel-probing} all the channel gain 
distributions are known a priori while here the reward distributions are revealed as and when the 
relays wake-up. We will discuss more about the work in \cite{chaporkar-proutiere08joint-probing}
in Section~\ref{general_class_section}.

%

Another work which is close to ours is that of Stadje \cite{stadje97two-levels}, where 
only some initial information  about an offer (e.g., the average size of 
the offer) is revealed to the decision maker upon its arrival. In addition to the actions, {stop} 
and {continue}, the decision maker can also choose to obtain more information about the offer by 
incurring a cost. Recalling previous offers is not allowed.
A similar problem is studied by Thejaswi et al.\ in \cite{thejaswi-etal10two-level-probing},
where initially a coarse estimate of the channel gain is made available to 
the transmitter. The transmitter can choose to probe the channel a second time to get a finer 
estimate. In both of these \cite{stadje97two-levels,thejaswi-etal10two-level-probing}, 
the optimal policy is characterized by a threshold rule.
However, the horizon length of these problems  is infinite, because of which the thresholds
are stage independent. In general, for a finite horizon problem the optimal policy would be stage 
dependent. For our problem, despite being a finite horizon one,
we are able to show that certain  stopping sets are identical across stages. This is
due to the fact that we  allow the best probed relay to stay awake.

\section{Restricted Class $\overline{\Pi}$: An MDP Formulation}
\label{restricted_class_section}
Confining to the restricted class $\overline{\Pi}$, in this section we will 
formulate the problem in (\ref{unconstrained_equn}) 
as a Markov decision process. This will require us to first discuss the 
one-step cost functions and state transitions before
proceeding to write the Bellman optimality equations. 

\subsection{One-Step Costs and State Transitions}
The decision instants or the decision stages are the times at which the relays wake-up. 
Thus, there are $N$ decision stages indexed by $k=1,2,\cdots,N$. 
Recall that for any policy in the restricted class $\overline{\Pi}$,
the decision at stage $k$ is based on $(b_k,H_k)$, where $b_k$ is the best reward so far 
and $H_k\in\mathcal{F}_k$ is the best reward distribution with $\mathcal{F}_k$ being the set of
reward distributions of all the unprobed relays so far. As mentioned earlier, if no relay has been
probed until stage $k$ then $b_k=-\infty$. On the other hand, if all the relays have been probed,
in which case $\mathcal{F}_k$ is empty, then we will denote the state as simple $b_k$.
Hence, the state space can be written as,
\begin{eqnarray*}
\mathcal{X} 
&=& [0,\overline{r}] \cup \Big\{(b,F_\ell): 
b\in\{-\infty\}\cup[0,\overline{r}], \ell\in\mathcal{L}\Big\}\cup\{\textbf{\emph{t}}\}
\end{eqnarray*}
where $\textbf{\emph{t}}$ is the cost-free termination state. We will use $(b,F_\ell)$ to denote
a generic state at stage $k$. 

Now, at stage $k=1,2,\cdots,N-1$, given that the state is $(b,F_\ell)$, if $\mathscr{F}$'s decision 
is to stop then the decision process enters $\textbf{\emph{t}}$, with $\mathscr{F}$ incurring a 
termination cost of $-\eta b$ (recall from (\ref{unconstrained_equn}) that $\eta>0$ is the 
trade-off parameter). On the other hand, if the action is to {continue} then $\mathscr{F}$ will 
first incur a waiting cost of $U_{k+1}$ (the time until the next relay wakes up) and then, when the 
$(k+1)$-th relay wakes-up (whose reward distribution is $F_{L_{k+1}}$),
$\mathscr{F}$ chooses between the two unprobed relays $-$ one the previous relay with reward 
distribution $F_\ell$, and other the new one with distribution $F_{L_{k+1}}$ $-$ so that 
the state at stage $k+1$ will be either $(b,F_{\ell})$ or $(b,F_{L_{k+1}})$. The best reward value 
continues to be $b$ since no new relay has been probed during the state transition.

Alternatively, $\mathscr{F}$ could choose the action to {probe} the available unprobed relay (whose 
reward distribution is $F_\ell$) incurring a cost of $\eta\delta$ (recall that $\delta$ is the 
probing cost). After probing, the decision process is still considered to be at stage $k$ with 
the new state being $b'=\max\{b,R_\ell\}$, where $R_\ell$ is the reward value of the just probed relay  
(thus the distribution of $R_\ell$ is $F_\ell$). $\mathscr{F}$ has to now further decide whether to 
{stop} (incurring a one-step cost of $-\eta b'$ and enter $\textbf{\emph{t}}$), or {continue} (in 
which case the one-step cost is $U_{k+1}$ and the next state is $(b',F_{L_{k+1}})$). 

Summarizing the above we can write the one-step cost, when the state 
at stage $k$ is $(b,F_\ell)$, as
\begin{eqnarray*}
g_k\Big((b,F_\ell),a_k\Big)&=&
\left\{\begin{array}{cl}
-\eta b & \mbox{ if } a_k=\st\\
U_{k+1} & \mbox{ if } a_k=\ct\\
\eta\delta & \mbox{ if } a_k=\pb.                                 
\end{array}\right.
\end{eqnarray*}
The next state, $X'$, is given by
\begin{eqnarray*}
X'&=&
\left\{\begin{array}{cl}
\textbf{\emph{t}} & \mbox{ if } a_k=\st\\
(b,F_\ell) \mbox{ or } (b,F_{L_{k+1}}) & \mbox{ if } a_k=\ct\\
\max\{b,R_\ell\} & \mbox{ if } a_k=\pb.                                 
\end{array}\right.
\end{eqnarray*}
We have used $X'$ to denote the next state instead of $X_{k+1}$ because, if $a_k=\pb$ then
the system is still at stage $k$. Only when the action is $\st$ or $\ct$ 
the system transits to the stage $k+1$.

Next, if the state at stage $k$ is $b$ (states of this form occur after probing the available 
unprobed relay; recall the above expressions when $a_k=\pb$), then
\begin{eqnarray*}
g_k(b,a_k)&=&
\left\{\begin{array}{cl}
-\eta b & \mbox{ if } a_k=\st\\
U_{k+1} & \mbox{ if } a_k=\ct,      
\end{array}\right. 
\end{eqnarray*}
and the next state is 
\begin{eqnarray*}
X_{k+1}&=&
\left\{\begin{array}{cl}
\textbf{\emph{t}} & \mbox{ if } a_k=\st\\
(b,F_{L_{k+1}}) & \mbox{ if } a_k=\ct.      
\end{array}\right. 
\end{eqnarray*}
The action to probe is not available whenever the state is $b$. 

At the last stage $N$, action $\ct$ is not available, so that 
\begin{eqnarray*}
g_N(b,F_\ell)&=&
\left\{\begin{array}{cl}
-\eta b & \mbox{ if } a_k=\st \\
\eta\delta & \mbox{ if } a_k=\pb,
\end{array}\right.
\end{eqnarray*}
with the system entering $\ts$ if $a_k=\st$, otherwise (i.e., if $a_k=\pb$) the state transits to $\max\{b,R_k\}$.
Finally, $g_N(b)=-\eta b$. 
Note that for a policy $\pi$, 
the expected sum of all the one-step costs
starting from stage $1$, plus the average waiting time for the first relay, 
$\mathbb{E}[U_1]=\tau$,\footnote{Since invariably a relay has to be chosen, every policy has to wait for at least the first 
relay to wake-up, at which instant the decision process begins. Thus, $U_1$ need not be accounted 
for in the total cost incurred by any policy.} 
will equal the total cost in (\ref{unconstrained_equn}).

\subsection{Cost-to-go Functions and the Bellman Equation}
Let $J_k$, $k=1,2,\cdots,N$, represent the optimal cost-to-go function at stage $k$.
Thus, $J_k(b)$ and $J_k(b,F_\ell)$ denote the cost-to-go, depending on 
whether there is, or is not an unprobed relay.
For the last stage, $N$, we have, $J_N(b)=-\eta b$, using which we obtain,
\begin{eqnarray}
\label{cost_to_go_stageN_statebF_equn}
J_N(b,F_\ell)&=&\min\Big\{-\eta b, \eta\delta + 
\mathbb{E}_\ell\Big[J_N(\max\{b,R_\ell\})\Big]\Big\}\nonumber\\
&=& \min\Big\{-\eta b, \eta\delta-\eta\mathbb{E}_\ell\Big[\max\{b,R_\ell\}\Big]\Big\},
\end{eqnarray}
where $\mathbb{E}_\ell[\cdot]$ denotes the expectation {with respect to} 
(w.r.t.) $R_\ell$ whose distribution is $F_\ell$. 
The first term in the $\min$-expression above is the cost of stopping
and the second term is the expected cost of probing and then stopping 
(recall that action $\ct$ is not available at the last stage $N$).
Next, for stages $k=1,2,\cdots,N-1$, denoting the expectation w.r.t.\
the distribution, $L$, of the location, $L_{k+1}$, of the next relay
by $\mathbb{E}_L[\cdot]$, we have
\begin{eqnarray}
\label{cost_to_go_stagek_stateb_equn}
J_k(b)&=&\min\Big\{-\eta b, \tau+\mathbb{E}_{L}\Big[J_{k+1}(b,F_{L_{k+1}})\Big]\Big\},
\end{eqnarray}
and
\begin{eqnarray}
\label{cost_to_go_stagek_statebF_equn}
J_k(b,F_\ell)
&=&\min\Big\{-\eta b, \eta\delta + \mathbb{E}_\ell\Big[J_k(\max\{b,R_\ell\})\Big], \nonumber\\
&&\hspace{-1cm} \tau + \mathbb{E}_{L}\Big[\min\{J_{k+1}(b,F_\ell),
J_{k+1}(b,F_{L_{k+1}})\}\Big]\Big\}.
\end{eqnarray}
The first term in both the min-expressions above is the cost of stopping. 
The middle term in (\ref{cost_to_go_stagek_statebF_equn}) is the expected cost of probing, with
$\eta\delta$ being the one-step cost and the remaining term being the future cost. 
The last term in both expressions is the expected cost of continuing, with $\tau$ 
representing the mean waiting time until the next relay wakes up.
The future cost-to-go in the last term of (\ref{cost_to_go_stagek_statebF_equn}) 
can be understood as follows.
When the state at stage $k=1,2,\cdots,N-1$ is $(b,F_\ell)$ and, 
if $\mathscr{F}$ decides to {continue}, then the reward distribution
of the next relay is $F_{L_{k+1}}$. Now, given the distributions $F_\ell$ and 
$F_{L_{k+1}}$, if $\mathscr{F}$ is asked to retain one of them, then it is optimal to go 
with the distribution that fetches a lower cost-to-go from
stage $k+1$ onwards, i.e., it is optimal to retain $F_\ell$ if 
$J_{k+1}(b,F_\ell)\le J_{k+1}(b,F_{L_{k+1}})$, otherwise retain $F_{L_{k+1}}$.\footnote{Formally 
one has to introduce an intermediate state of the form 
$(b,F_\ell,F_{L_{k+1}})$ at stage $k+1$
where the only actions available are, choose $F_\ell$ or $F_{L_{k+1}}$. 
Then $J_{k+1}(b,F_\ell,F_{L_{k+1}})=\min\{J_{k+1}(b,F_\ell),J_{k+1}(b,F_{L_{k+1}})\}$, which, for 
simplicity, we are directly using in (\ref{cost_to_go_stagek_statebF_equn}).}
Later in this section we will show that, given two distributions, $F_\ell$ and $F_u$, 
if $F_\ell$ is \emph{stochastically greater than}  $F_u$ \cite{stoyan83comparison-methods-queues} 
then $J_{k+1}(b,F_\ell)\le J_{k+1}(b,F_u)$  (see Lemma~\ref{cost_to_go_ordering_lemma}-(i))  
so that it is optimal to retain the stochastically greater distribution.

First, for simplicity let us introduce the following notation. For $k=1,2,\cdots,N-1$, 
let $C_k$ represent the 
cost of continuing:
\begin{eqnarray}
\label{continuing_cost_b_equn}
C_k(b)&=&\tau+\mathbb{E}_{L}\Big[J_{k+1}(b,F_{L_{k+1}})\Big]
\end{eqnarray}
\begin{eqnarray}
\label{continuing_cost_bF_equn}
C_k(b,F_\ell)
= \tau + \mathbb{E}_{L}\Big[\min\{J_{k+1}(b,F_\ell),J_{k+1}(b,F_{L_{k+1}})\}\Big].
\end{eqnarray}
For $k=1,2,\cdots,N$, the cost of probing, $P_k$, is given by
\begin{eqnarray}
\label{probing_cost_bF_equn}
P_k(b,F_\ell)&=&\eta\delta + \mathbb{E}_\ell\Big[J_k(\max\{b,R_\ell\})\Big].
\end{eqnarray}
From (\ref{continuing_cost_b_equn}) and (\ref{continuing_cost_bF_equn}) it is 
immediately clear that $C_k(b,F_\ell)\le C_k(b)$
for any $F_\ell$ ($\ell\in\mathcal{L}$). This inequality should be intuitive as well, since
$\mathscr{F}$ can expect to accrue a better cost if, in addition to a probed 
relay, it also possesses an unprobed relay.
It will be useful to note this inequality as a lemma.
\begin{lemma}
\label{continuing_cost_corollary}
For $k=1,2,\cdots,N-1$ and any $(b,F_\ell)$ we have $C_k(b,F_\ell)\le C_k(b)$. 
\end{lemma}
\begin{IEEEproof}
As discussed just before the Lemma statement, the inequality follows 
easily from the expressions of these costs;
recall (\ref{continuing_cost_b_equn}) and (\ref{continuing_cost_bF_equn}).
\end{IEEEproof}

Finally, using the above cost notation, the cost-to-go 
functions in (\ref{cost_to_go_stagek_stateb_equn}) 
and (\ref{cost_to_go_stagek_statebF_equn}) 
can be written as, for $k=1,2,\cdots,N-1$,
\begin{eqnarray}
\label{Jk_b_equn}
J_k(b)&=&\min\Big\{-\eta b, C_k(b)\Big\}\\
\label{Jk_bF_equn}
J_k(b,F_\ell)&=&\min\Big\{-\eta b,  P_k(b,F_\ell), C_k(b,F_\ell)\Big\}.
\end{eqnarray}

\subsection{Ordering Results for the Cost-to-go Functions}
We will examine how the cost-to-go functions $J_k(b)$ and $J_k(b,F_\ell)$ behave as functions of
 $F_\ell$ and the stage index $k$. 
We will first require the definition of stochastic ordering.
\begin{definition}[Stochastic Ordering]
\label{stochastic_ordering_definition}
Given two distributions $F_\ell$ and $F_u$, $F_\ell$ is stochastically greater than
$F_u$, denoted as $F_\ell\ge_{st}F_u$, if $1-F_\ell(r)\ge 1-F_u(r)$, for all $r$. Equivalently 
\cite{stoyan83comparison-methods-queues}, $F_\ell\ge_{st}F_u$ if and only 
if for every non-decreasing function $f:\Re\rightarrow\Re$,
$\mathbb{E}_\ell[f(R_\ell)]\ge\mathbb{E}_u[f(R_u)]$ where 
the distributions of $R_\ell$ and $R_u$ are $F_\ell$ and $F_u$, respectively.
\hfill $\blacksquare$
\end{definition}

Now, consider two relays at locations $\ell$ and $u$. If the corresponding
reward distributions, $F_\ell$ and $F_u$, are such that 
$F_\ell\ge_{st} F_u$ then $\mathscr{F}$ can expect that probing the relay at $\ell$
would yield a better reward value than the relay at $u$. Thus, $\mathscr{F}$ would prefer
the stochastically greater reward distribution $F_\ell$, over $F_u$. 
Extending this observation, it is reasonable to expect that 
$\mathscr{F}$ can accrue lower expected costs (total, continuing and probing costs) 
if the unprobed reward distribution available at stage $k$ is $F_\ell$
than if it is $F_u$. We will formally prove this result next. Also, we will show that 
the expected cost at stage $k$ is less than that at stage $k+1$, i.e.,
$J_k(x)\le J_{k+1}(x)$ for any state $x$. This again should be intuitive because,
starting from stage $k$, $\mathscr{F}$ has the option to observe an additional relay 
than if it were to start from stage $k+1$. With more resource available, and 
with these being i.i.d., $\mathscr{F}$ should achieve a better cost. 
We will state these two results in the following lemma.

\begin{lemma}
\label{cost_to_go_ordering_lemma}
\begin{enumerate}
\item[(i)] For $k=1,2,\cdots,N-1$, if $F_\ell\ge_{st}F_u$ then $C_k(b,F_\ell)\le C_k(b,F_u)$, 
(and including $k=N$) $P_k(b,F_\ell)\le P_k(b,F_u)$ and $J_{k}(b,F_\ell)\le J_{k}(b,F_u)$.
\item[(ii)] For $k=1,2,\cdots,N-2$, $C_k(b)\le C_{k+1}(b)$ and $C_k(b,F_\ell)\le C_{k+1}(b,F_\ell)$, 
(and including $k=N-1$) $P_k(b,F_\ell)\le P_{k+1}(b,F_\ell)$ and $J_k(b,F_\ell)\le J_{k+1}(b,F_\ell)$.
\end{enumerate}
\end{lemma}
\begin{IEEEproof}
To prove (i) we first show that the various costs are non-increasing functions of $b$.
We then complete the proof using the definition of stochastic ordering 
(Definition~\ref{stochastic_ordering_definition}).
Part~(ii) follows from induction.   
Detail proofs are available in Appendix~\ref{cost_to_go_ordering_lemma_appendix}.
\end{IEEEproof}


\section{Restricted Class $\overline{\Pi}$: Structural Results}
\label{structural_results_section}
We begin by defining, at stage $k=1,2,\cdots,N-1$, the \emph{stopping set} $\mathcal{S}_k$ as
\begin{eqnarray}
\label{optimal_stopping_equn}
\mathcal{S}_k
&=& \Big\{b: -\eta b\le C_k(b)\Big\}.
\end{eqnarray}
From (\ref{Jk_b_equn}) it follows that the stopping set $\mathcal{S}_k$ is the set 
of all states $b$ (states of this form are obtained after probing at stage $k$) 
where it is better to {stop} than to {continue}.

Similarly, for a given
distribution $F_\ell$ we define the {stopping set} $\mathcal{S}_k^\ell$ as, for $k=1,2,\cdots,N-1$,
\begin{eqnarray}
\label{optimal_stopping_l_equn}
\mathcal{S}_k^\ell
&=&\Big\{b:-\eta b\le \min\{P_k(b,F_\ell), C_k(b,F_\ell)\} \Big\}.
\end{eqnarray}
Using (\ref{Jk_bF_equn}) the set $\mathcal{S}_{k}^\ell$
has to be interpreted as, for a given distribution $F_\ell$, the set of $b$ such that whenever the 
state at stage $k$ is $(b,F_\ell)$
it is better to {stop} than to either {probe} or {continue}. 
Note that when $b=-\infty$ it is never optimal to stop; hence, both these stopping sets
are subsets of $[0,\overline{r}]$.
Finally, stopping sets can also be
defined for $k=N$ as, $\mathcal{S}_N=[0,\overline{r}]$ (since, at the last stage $N$,
for any $b$ the only action available is to stop), and
\begin{eqnarray}
\label{optimal_stopping_lN_equn}
\mathcal{S}_N^\ell
&=& \Big\{b:-\eta b\le P_k(b,F_\ell) \Big\}.
\end{eqnarray}

The following set inclusion properties easily follow from the definition of these sets 
and the properties of the cost functions in Lemma~\ref{continuing_cost_corollary} and
Lemma~\ref{cost_to_go_ordering_lemma}.
\begin{lemma}
\label{sets_trivial_ordering_corollary}
\verb11
\begin{enumerate}
\item[(i)] For $k=1,2,\cdots,N$ and for any $F_\ell$ we have
$\mathcal{S}_k^\ell\subseteq\mathcal{S}_k$.
\item[(ii)] For $k=1,2,\cdots,N$, if $F_\ell\ge_{st}F_u$ then
$\mathcal{S}_k^\ell\subseteq\mathcal{S}_k^u$.
\item[(iii)] For $k=1,2,\cdots,N-1$ we have
$\mathcal{S}_k\subseteq\mathcal{S}_{k+1}$,
and for any $F_\ell$, $\mathcal{S}_k^\ell\subseteq\mathcal{S}_{k+1}^\ell$.
\end{enumerate}
\end{lemma}
\begin{IEEEproof}
Recall the definition of the stopping sets from  (\ref{optimal_stopping_equn}) and
(\ref{optimal_stopping_l_equn}). Part~(i) follows from Lemma~\ref{continuing_cost_corollary}.
Parts~(ii) and (iii) are due to 
Parts~(i) and (ii) of Lemma~\ref{cost_to_go_ordering_lemma}, respectively.
\end{IEEEproof}

\emph{Discussion:} The above results can be understood as follows. 
Whenever an unprobed relay (say with reward distribution $F_\ell$) is available, $\mathscr{F}$
can be more stringent about the best reward values, $b$, for which it chooses to stop.
This is because, $\mathscr{F}$ can now additionally choose to 
probe $F_\ell$ possibly yielding a better reward than $b$. 
Thus, unless the best reward $b$ is already good (so that there is no gain in probing $F_\ell$), 
$\mathscr{F}$ will not choose to stop. Hence, we have $\mathcal{S}_k^\ell\subseteq\mathcal{S}_k$.
Next, if $F_\ell\ge_{st}F_u$ then since probing $F_\ell$ has a higher chance of yielding a better
reward, the stopping condition is more stringent if the reward 
distribution of the available unprobed relay 
is $F_\ell$ than $F_u$. Hence, the corresponding stopping 
sets are ordered as in Part~(ii) of the above lemma, i.e., 
$\mathcal{S}_k^\ell\subseteq\mathcal{S}_k^u$. Finally, whenever there are more stages to-go,
$\mathscr{F}$ can be more cautious about stopping since it has the 
option to observe more relays. This suggests that  $\mathcal{S}_k\subseteq\mathcal{S}_{k+1}$ 
and $\mathcal{S}_k^\ell\subseteq\mathcal{S}_{k+1}^\ell$.

From our above discussion, the phrase ``$\mathscr{F}$ being more stringent about stopping,''
suggests that it may be better to stop for larger values of $b$.
Equivalently, this would mean that the stopping sets are characterized by \emph{thresholds},
beyond which it is optimal to stop. This is exactly our first main result
(Theorem~\ref{threshold_nature_lemma}). 
Later we will prove a more interesting result 
(Theorem~\ref{stopping_sets_equal_theorem} and \ref{stopping_sets_l_equal_theorem})
where we show that the stopping sets are \emph{stage independent}, 
i.e., $\mathcal{S}_k=\mathcal{S}_{k+1}$ and $\mathcal{S}_k^\ell=\mathcal{S}_{k+1}^\ell$. 
In the following sub-sections we will work the details of these two results.

\vspace{-4mm}
\subsection{Stopping Sets: Threshold Property}
To prove the threshold structure of the stopping sets the following key lemma is required where we
show that the increments in the various costs are bounded by the increments in the cost of stopping.
\begin{lemma}
\label{costs_bounded_lemma}
For $k=1,2,\cdots,N-1$ (for Part~(ii), $k=1,2,\cdots,N$), for any $F_\ell$, 
and for $b_2>b_1$ we have
\begin{enumerate}
\item[(i)] $C_k(b_1)-C_k(b_2)\le\eta(b_2-b_1)$,
\item[(ii)] $P_k(b_1,F_\ell)-P_k(b_2,F_\ell)\le\eta(b_2-b_1)$ 
\item[(iii)] $C_k(b_1,F_\ell)-C_k(b_2,F_\ell)\le\eta(b_2-b_1)$. 
\end{enumerate}
\end{lemma}
\begin{IEEEproof}
Available in Appendix~\ref{threshold_nature_lemma_appendix}.
\end{IEEEproof}

\begin{theorem}
\label{threshold_nature_lemma}
For $k=1,2,\cdots,N$ and for $b_2>b_1$,
\begin{enumerate}
 \item[(i)] If $b_1\in\mathcal{S}_k$ then $b_2\in\mathcal{S}_k$.
 \item[(ii)] For any $F_\ell$, if $b_1\in\mathcal{S}_k^\ell$ then $b_2\in\mathcal{S}_k^\ell$.
\end{enumerate}
\end{theorem}
\begin{IEEEproof} 
Since $\mathcal{S}_N=[0,\overline{r}]$, Part~(i) trivially holds for $k=N$. 
Next, for $k=1,2,\cdots,N-1$,
using Lemma~\ref{costs_bounded_lemma}-(i) we can write,
\begin{eqnarray*}
-\eta b_2 \le -\eta b_1 - C_k(b_1) + C_k(b_2).
\end{eqnarray*}
Since $b_1\in\mathcal{S}_k$, from (\ref{optimal_stopping_equn}) 
we know that $-\eta b_1\le C_k(b_1)$,
using which in the above expression we obtain $-\eta b_2\le C_k(b_2)$ 
implying that $b_2\in\mathcal{S}_k$.
\emph{Part~(ii)} can be similarly completed using 
Parts~(ii) and (iii) of Lemma~\ref{costs_bounded_lemma}.
\end{IEEEproof}

\begin{figure}[h]
\centering
\includegraphics[scale=0.8]{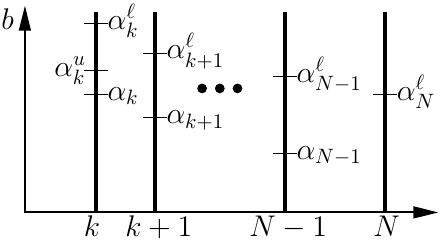}
\caption{\label{threshold_property_figure}
Illustration of the threshold property: the vertical lines are the reward axis, with 
each line corresponding to a different stage. The stopping sets are represented 
by marking their thresholds on the respective vertical lines.}
\vspace{-2mm}
\end{figure}

\emph{{Discussion:}} Thus, the stopping sets $\mathcal{S}_k$ and $\mathcal{S}_k^\ell$ can be 
characterized in terms of lower bounds $\alpha_k$ and $\alpha_k^\ell$,
respectively, as illustrated in Fig.~\ref{threshold_property_figure} (see the vertical line
corresponding to the stage index $k$). 
Also shown in Fig.~\ref{threshold_property_figure} is the threshold, 
$\alpha_k^u$, corresponding to a distribution $F_u\le_{st}F_\ell$. From 
Lemma~\ref{sets_trivial_ordering_corollary}-(i) and \ref{sets_trivial_ordering_corollary}-(ii) 
it follows that these thresholds are ordered, $\alpha_k\le\alpha_k^u\le\alpha_k^\ell$.
Further, in Fig.~\ref{threshold_property_figure} we have depicted these thresholds to be 
decreasing with the stage index $k$ (vertical lines from left to right); this is due to
Lemma~\ref{sets_trivial_ordering_corollary}-(iii) from where we know that the stopping sets
are increasing with $k$. 
Our main result in the next section 
(Theorem~\ref{stopping_sets_equal_theorem} and \ref{stopping_sets_l_equal_theorem}) is 
to show that these thresholds are, in fact, equal (i.e., $\alpha_k=\alpha_{k+1}$ 
and $\alpha_k^\ell=\alpha_{k+1}^\ell$). 
Finally, note that in 
Fig.~\ref{threshold_property_figure} we have
not shown the threshold $\alpha_N$ corresponding to the stopping set $\mathcal{S}_N$; this is
simply because $\alpha_N=0$ (since $\mathcal{S}_N=[0,\overline{r}]$).

%

\subsection{Stopping Sets: Stage Independence Property}
From Lemma~\ref{sets_trivial_ordering_corollary}-(iii) we already know that 
$\mathcal{S}_k\subseteq\mathcal{S}_{k+1}$, and $\mathcal{S}_k^\ell\subseteq\mathcal{S}_{k+1}^\ell$.
In this section we will prove the inclusion in the other direction, thus leading to the result that 
the stopping sets are identical across the stages. We will begin by defining the
{\textbf{\emph{stobing}} (\textbf{\emph{sto}}pping-or-pro\textbf{\emph{bing}}) set} 
$\mathcal{Q}_k^\ell$ as, for $k=1,2,\cdots,N-1$,
\begin{eqnarray}
\label{optimal_stopping_probing_equn}
\mathcal{Q}_{k}^\ell
&=& \Big\{b: \min\{-\eta b, P_k(b,F_\ell)\} \le C_k(b,F_\ell)\Big\}.
\end{eqnarray} 
From (\ref{Jk_bF_equn}) it follows that $\mathcal{Q}_{k}^\ell$ is,
for a given distribution $F_\ell$, the set of all $b$ such that whenever the state at stage $k$
is $(b,F_\ell)$
it is better to either {stop} or {probe} than to {continue}.
From the definition of the sets $\mathcal{S}_k^\ell$ and $\mathcal{Q}_{k}^\ell$ (in 
(\ref{optimal_stopping_l_equn}) and (\ref{optimal_stopping_probing_equn}), 
respectively) it immediately
follows that $\mathcal{S}_k^\ell\subseteq\mathcal{Q}_{k}^\ell$. Also from 
Lemma~\ref{sets_trivial_ordering_corollary}-(i) we already know that
$\mathcal{S}_k^\ell\subseteq\mathcal{S}_k$.
However, it is not immediately clear how the sets $\mathcal{Q}_k^\ell$ and $\mathcal{S}_k$
are ordered. We will show that if $\mathcal{F}=\{F_\ell:\ell\in\mathcal{L}\}$ is \emph{totally 
stochastically ordered} (to be defined next) then $\mathcal{S}_k\subseteq\mathcal{Q}_k^\ell$ 
(Lemma~\ref{all_distributions_corollary}). This result is essential for proving our main theorems.

\begin{definition}[Total Stochastic Ordering]
\label{total_stochastic_ordering_definition}
$\mathcal{F}$ is said 
to be \emph{totally stochastically ordered} if any two distributions from 
$\mathcal{F}$ are stochastically ordered.
Formally, for any $F_\ell,F_u\in\mathcal{F}$ either $F_\ell\ge_{st}F_u$ or $F_u\ge_{st}F_\ell$.
Further, if there exists a {distribution} $F_m\in\mathcal{F}$
such that for every $F_\ell\in\mathcal{F}$ we have $F_\ell\ge_{st}F_m$ then
we say that $\mathcal{F}$
is \emph{totally stochastically ordered with a minimum distribution}.
\hfill $\blacksquare$
\end{definition}

\begin{lemma}
\label{F_total_order_lemma}
The set of reward distributions $\mathcal{F}$ in (\ref{distribution_set}), is 
totally stochastically ordered with a minimum distribution. 
\end{lemma}
\begin{IEEEproof}
The channel gains, $\{G_\ell:\ell\in\mathcal{L}\}$, being identically distributed will be
essential to show that $\mathcal{F}$ is totally stochastically ordered. 
Existence of a minimum distribution will require the assumption we had
made earlier (in Section~\ref{local_forwarding_problem_section}) that
$\mathcal{L}$ is compact (closed and bounded). The
complete proof is available in Appendix~\ref{F_total_order_lemma_appendix}.
\end{IEEEproof}

{\emph{Remark:}}
Our subsequent results are not simply limited to the
$\mathcal{F}$ in (\ref{distribution_set}) which is the distribution set
arising from the particular reward structure, $R_\ell$, we had assumed in (\ref{reward_equn}).
One can consider any collection of bounded reward random variables $\{R_\ell\}$, 
such that the corresponding $\mathcal{F}$ is 
totally stochastically ordered with a minimum distribution, still all the 
subsequent results will hold.

Before proceeding to our main theorems, we need the following results.
\begin{lemma}
\label{equal_costs_lemma}
Suppose $\mathcal{S}_k\subseteq\mathcal{Q}_k^u$, for some $F_u$, and some $k=1,2,\cdots,N-1$. 
Then for every $b\in\mathcal{S}_k$ we have $J_k(b,F_u)=J_N(b,F_u)$.
\end{lemma}
\begin{IEEEproof}
Available in Appendix~\ref{equal_costs_lemma_appendix}. 
\end{IEEEproof}

Next we show that the hypothesis in the above lemma indeed holds for 
every $F_\ell\in\mathcal{F}$.
\begin{lemma}
\label{all_distributions_corollary}
For $k=1,2,\cdots,N-1$ and for any $F_\ell\in\mathcal{F}$ we have
$\mathcal{S}_k\subseteq\mathcal{Q}_k^\ell$.
\end{lemma}
\begin{IEEEproof}
The proof involves two steps:

1)  First we show that 
if there exists an $F_u$ such that, for $k=1,2,\cdots,N-1$, $\mathcal{S}_k\subseteq\mathcal{Q}_k^u$ 
(thus satisfying the hypothesis in Lemma~\ref{equal_costs_lemma}), then for every $F_\ell\ge_{st} F_u$  
we have $\mathcal{S}_k\subseteq\mathcal{Q}_k^\ell$. Lemma~\ref{equal_costs_lemma} and the 
total stochastic ordering of $\mathcal{F}$ are required for this part.

2) Next we show that 
a minimum distribution $F_m$ satisfies the hypothesis in Lemma~\ref{equal_costs_lemma}, 
i.e., for every $k=1,2,\cdots,N-1$, $\mathcal{S}_k\subseteq\mathcal{Q}_k^m$.
The proof is completed by recalling that $F_\ell\ge_{st} F_m$ for every $F_\ell\in\mathcal{F}$ 
and then using in \emph{Step~1}, $F_m$
in the place of $F_u$. The existence of a minimum 
distribution $F_m$ (recall Lemma~\ref{F_total_order_lemma}) is essential here. 

Formal proofs of both steps are available in Appendix~\ref{all_distributions_corollary_appendix}.
\end{IEEEproof}

\noindent
The following are the main theorems of this section:
\begin{theorem}
\label{stopping_sets_equal_theorem}
 For $k=1,2,\cdots,N-2$, $\mathcal{S}_k=\mathcal{S}_{k+1}$.
\end{theorem}
\begin{IEEEproof}
From Lemma~\ref{sets_trivial_ordering_corollary}-(iii) we already know that
$\mathcal{S}_k\subseteq\mathcal{S}_{k+1}$. 
Here, we will show that $\mathcal{S}_k\supseteq\mathcal{S}_{k+1}$. 
Fix a $b\in\mathcal{S}_{k+1}\subseteq\mathcal{S}_{k+2}$.
From Lemma~\ref{all_distributions_corollary} we know that
$\mathcal{S}_{k+1}\subseteq\mathcal{Q}_{k+1}^\ell$
and $\mathcal{S}_{k+2}\subseteq\mathcal{Q}_{k+2}^\ell$,
for every $F_\ell$. Now, applying Lemma~\ref{equal_costs_lemma} we can write, 
$J_{k+1}(b,F_\ell)=J_{k+2}(b,F_\ell)=J_N(b,F_\ell)$.
Thus,
\begin{eqnarray*}
C_{k+1}(b)&=& \tau + \mathbb{E}_L\Big[J_{k+2}(b,F_{L_{k+2}})\Big]\nonumber\\
&=& \tau + \mathbb{E}_L\Big[J_{k+1}(b,F_{L_{k+1}})\Big]\nonumber\\
&=& C_k(b)
\end{eqnarray*}
Finally, since $b\in\mathcal{S}_{k+1}$ we  have $-\eta b\le C_{k+1}(b)= C_k(b)$
which implies that $b\in\mathcal{S}_k$.
\end{IEEEproof}

\emph{{Discussion:}} It is interesting to compare the above
result with the solution obtained for the basic model 
(i.e., $\delta=0$ case; recall the discussion on related work in
Section~\ref{local_forwarding_problem_section})
or equivalently the basic asset selling problem 
\cite[Section~4.4]{bertsekas05optimal-control-vol1}.
In \cite[Section~4.4]{bertsekas05optimal-control-vol1}, as in our
Theorem~\ref{stopping_sets_equal_theorem}
here, it is shown that similar 
stopping sets are identical across the stages; this policy is referred to as the 
\emph{one-step-look-ahead} rule since the policy, to stop if and only if
the ``cost of stopping'' is less than the ``cost of continuing for one-more step and 
then stopping,'' being optimal for  stage $N-1$, is optimal for all stages. 
The key idea there (i.e., in \cite[Section~4.4]{bertsekas05optimal-control-vol1}),
as in our Lemma~\ref{equal_costs_lemma}, is also to show that 
the cost-to-go functions, at every stage $k$, are identical for every 
state within the stopping set.
However here, to apply Lemma~\ref{equal_costs_lemma}, it was further essential 
for us to prove Lemma~\ref{all_distributions_corollary}
showing that for every $F_\ell$, $\mathcal{S}_k\subseteq\mathcal{Q}_k^\ell$.
Now, note that the result $\mathcal{S}_k\subseteq\mathcal{Q}_k^\ell$ trivially holds
for $\delta = 0$, since if  $\delta = 0$ then for any $(b,F_\ell)$ it is always optimal to probe, 
so that $\mathcal{Q}_k^\ell=[0,\overline{r}]$.
Thus, Theorem~\ref{stopping_sets_equal_theorem}, incorporating the additional case $\delta>0$,
can be considered as a generalization of the one-step-look-ahead rule 
which is optimal for the basic asset selling model.

\begin{theorem}
\label{stopping_sets_l_equal_theorem}
For $k=1,2,\cdots,N-1$ and  any $F_\ell$,
$\mathcal{S}_k^\ell=\mathcal{S}_{k+1}^\ell$.
\end{theorem}
\begin{IEEEproof}
Similar to the proof of Theorem~\ref{stopping_sets_equal_theorem},
 here 
we need to show that
the probing and continuing costs satisfy analogous equalities, i.e., 
for $b\in\mathcal{S}_k^\ell$ we need to show that
$P_{k+1}(b,F_\ell)=P_{k}(b,F_\ell)$ and 
$C_{k+1}(b,F_\ell)=C_{k}(b,F_\ell)$. 
Formal proof is available in Appendix~\ref{stopping_sets_l_equal_theorem_appendix}.
\end{IEEEproof}

\begin{figure}[t]
\centering
\includegraphics[scale=0.8]{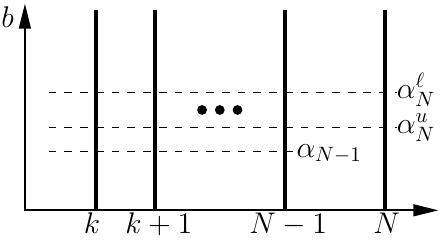}
\vspace{-2mm}
\caption{\label{stage_independence_figure}
Illustration of the stage independence property: only the thresholds corresponding to the 
last stage (and stage $N-1$ for $\mathcal{S}_k$) are shown, since
these alone are sufficient to characterize the stopping sets for any $k$.}
\vspace{-6mm}
\end{figure}

\emph{Discussion:} Owing to Theorem~\ref{stopping_sets_equal_theorem} and
\ref{stopping_sets_l_equal_theorem}, we can now modify the illustration in
Fig.~\ref{threshold_property_figure} to Fig.~\ref{stage_independence_figure}
where we show only a single threshold corresponding to each stopping set. 
Thus, to characterize the stopping set 
$\mathcal{S}_k^\ell$ for any $k$, it is sufficient to compute only the threshold 
$\alpha_N^\ell$ corresponding to the last stage. Similarly, the stopping set 
$\mathcal{S}_k$ is characterized by the threshold $\alpha_{N-1}$ computed for stage 
$N-1$ (recall that $\alpha_N=0$).

\vspace{-3mm}
\subsection{Probing Sets}
\label{probing_sets_section}
Similar to the stopping sets $\mathcal{S}_k^\ell$, one can also 
define the {probing sets} $\mathcal{P}_k^\ell$
as the set of all $b$ such that whenever the state at stage $k$ is $(b,F_\ell)$ 
it is better to {probe} than to either {stop} or {continue}, i.e.,
\begin{eqnarray}
\label{optimal_probing_set_equn}
\mathcal{P}_k^\ell
&=&\Big\{b: P_k(b,F_\ell)\le\min\{-\eta b, C_k(b,F_\ell)\}\Big\}.
\end{eqnarray}
Note that $\mathcal{P}_k^\ell$ is simply the difference of the sets $\mathcal{Q}_k^\ell$
and $\mathcal{S}_k^\ell$, i.e., $\mathcal{P}_k^\ell=\mathcal{Q}_k^\ell\setminus\mathcal{S}_k^\ell$.

From our numerical work we have observed that, similar to the stopping sets,  the probing 
sets $\mathcal{P}_k^\ell$ are characterized by upper bounds $\zeta_k^\ell$ (see 
Fig.~\ref{probing_conjecture_figure}).
The intuition for this
is as follows. Let $(b,F_\ell)$ be the state at stage $N-1$. If
the value of $b$ is very small, then it is better to probe than to continue, because 
probing will give an opportunity to probe an additional relay at 
stage $N$ in case the process continues after probing at stage $N-1$,
while continuing without probing
will deprive $\mathscr{F}$ of this 
opportunity. This argument can be extended to any stage $k$ to conclude that it may be better
to probe for small values of $b$. However, as $b$ increases, probing may not yield a better
reward than the existing $b$;
hence probing might not be worth the cost, so that
it may be better to simply continue. 

To formally show the threshold property of the probing set $\mathcal{P}_k^\ell$,
the following is sufficient: for any $b_2>b_1$, 
\begin{eqnarray*}
P_k(b_1,F_\ell)-P_k(b_2,F_\ell)
&\le& C_k(b_1,F_\ell)-C_k(b_2,F_\ell).
\end{eqnarray*}
This is because, if $b_2\notin\mathcal{S}_k^\ell$ (so that stopping is not optimal) is such that
$b_2\in\mathcal{P}_k^\ell$ (i.e., $P_k(b_2,F_\ell)\le C_k(b_2,F_\ell)$) then
from the above inequality we obtain  $P_k(b_1,F_\ell)\le C_k(b_1,F_\ell)$,
implying that it is optimal to probe at $b_1$ as well so that probing sets 
are characterized by upper bounds. 
However, we have not yet been able to 
prove or disprove such a result,
but we strongly believe that it is true and make the following conjecture.
\begin{conjecture}
\label{probing_set_conjecture}
For $k=1,2,\cdots,N-1$, for any $F_\ell$, if $b_2\in\mathcal{P}_k^\ell$ 
then for any $b_1<b_2$ we have $b_1\in\mathcal{P}_k^\ell$.
\hfill $\blacksquare$
\end{conjecture}

\begin{figure}
\centering
\subfigure[]{
\includegraphics[scale=0.8]{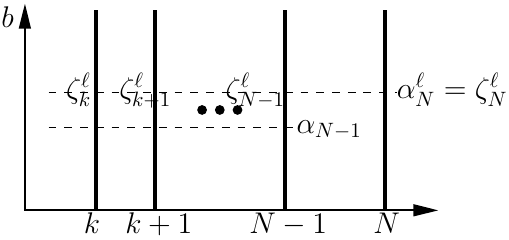}
\label{probingF_ell_figure}
}
\subfigure[]{
\includegraphics[scale=0.8]{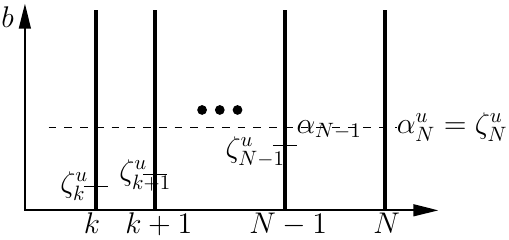}
\label{probingF_u_figure}
}
\vspace{-2mm}
\caption{\label{probing_conjecture_figure}
Structure of the probing sets if Conjecture~\ref{probing_set_conjecture}
is true. \subref{probingF_ell_figure} Probing sets corresponding to a distribution
$F_\ell$ such that $\alpha_N^\ell>\alpha_{N-1}$, 
\subref{probingF_u_figure} Probing sets corresponding to an $F_u$ such that
$\alpha_N^u=\alpha_{N-1}$}
\vspace{-6mm}
\end{figure}

\emph{Discussion:} If the above conjecture is true, then some additional structural results
can be deduced. For instance, suppose for some $F_\ell$, $\alpha_k^\ell>\alpha_k$, or 
equivalently, $\alpha_N^\ell>\alpha_{N-1}$ (refer to Fig.~\ref{probingF_ell_figure}).
Then, since $\mathcal{S}_k\subseteq\mathcal{Q}_k^\ell$ (from 
Lemma~\ref{all_distributions_corollary}), for any $(b,F_\ell)$ such that 
$\alpha_{N-1}<b<\alpha_N^\ell$, it should be optimal to probe. Now, invoking 
Conjecture~\ref{probing_set_conjecture} we can conclude that it is optimal to probe
for any $b<\alpha_N^\ell$, so that $\zeta_k^\ell=\alpha_N^\ell$ for all $k$.
Thus, for such ``good'' distributions, $F_\ell$, (i.e., $F_\ell$ such that
$\alpha_N^\ell>\alpha_{N-1}$) the policy corresponding to it is completely characterized
by a single threshold $\alpha_{N}^\ell$. Next, for distributions $F_u$ 
such that $\alpha_k^u=\alpha_k$ (equivalently, $\alpha_N^\ell=\alpha_{N-1}$;
see Fig.~\ref{probingF_u_figure}), there is a window between $\zeta_k^u$
and $\alpha_N^u$ where, for any $(b,F_\ell)$ such that 
$\zeta_N^u\le b<\alpha_N^u$, it is optimal to continue.
Unlike $\alpha_k^u$, the thresholds $\zeta_k^u$ are stage dependent. 
In fact, from our numerical work, we
observe that $\zeta_k^u$ are increasing with $k$. Finally, as depicted in 
Fig.~\ref{probing_conjecture_figure}, for any distribution
$F_\ell$, at the last stage we invariably should have $\alpha_N^\ell=\zeta_N^\ell$ since
the action to continue is not available at stage $N$.

\begin{figure*}[b]
\vspace{-6mm}
\begin{eqnarray}
\label{bellman_gen_last_equn}
J_N(b,\mathcal{H})
&=&\min\Big\{-\eta b, \eta\delta + \min_{F_\ell\in\mathcal{H}}\mathbb{E}_{\ell}
\Big[J_N(\max\{b,R_\ell\}, \mathcal{H}\setminus\{F_\ell\})\Big] \Big\}.
\end{eqnarray}
\vspace{-4mm}
\begin{eqnarray}
\label{bellman_gen_k_G_equn}
J_k(b,\mathcal{H}) 
&=& \min\Big\{-\eta b, \eta\delta + 
\min_{F_\ell\in\mathcal{H}}\mathbb{E}_{\ell}\Big[J_k(\max\{b,R_\ell\},
\mathcal{H}\setminus\{F_\ell\})\Big], 
\tau + \mathbb{E}_L\Big[J_{k+1}(b,\mathcal{H}\cup\{F_{L_{k+1}}\})\Big]\Big\}.
\end{eqnarray}
\end{figure*}

\vspace{-4mm}
\subsection{Policy Implementation}
\label{policy_implementation_section}
To summarize, from Theorem~\ref{threshold_nature_lemma}, the stopping sets $\mathcal{S}_k$ and 
$\mathcal{S}_k^\ell$ 
are characterized by 
lower bounds $\alpha_k$ and $\alpha_k^\ell$. In Theorem~\ref{stopping_sets_equal_theorem}
and \ref{stopping_sets_l_equal_theorem} 
we proved that these thresholds are stage independent. Hence it is sufficient to 
compute only $\alpha_{N-1}$ and $\alpha_{N}^\ell$, thus simplifying the overall 
computation of the optimal policy.
Further, if Conjecture~\ref{probing_set_conjecture} is true, then the 
upper bounds $\zeta_k^\ell$ are sufficient
to characterize the probing sets $\mathcal{P}_k^\ell$. 

Now, $\mathscr{F}$ after computing these thresholds,
operates as follows: At stage $k=1,2,\cdots,N-1$,
whenever the state is $(b,F_\ell)$,
(\textbf{1}) if $b \ge \alpha_{N}^\ell$ then {stop} and forward the packet
to the probed relay,
(\textbf{2}) if $b \le \zeta_k^\ell$ then {probe} the unprobed relay and update the
best reward to $b'=\max\{b,R_\ell\}$. Now, if $b'\ge\alpha_{N-1}$ {stop}, 
otherwise {continue} to wait for the next relay,
(\textbf{3}) otherwise (i.e., if $\zeta_k^\ell < b < \alpha_{N}^\ell$), 
{continue} to wait for the 
next relay to wake-up, at which instant choose, between $F_\ell$ and $F_{L_{k+1}}$, 
whichever is stochastically greater while putting the other unprobed relay to sleep.

If the decision process enters the last stage $N$ and if the state is $(b,F_\ell)$ then 
 if $b\ge\alpha_N^\ell$ stop, otherwise probe
(continue is not available). Finally, if the state at stage $N$ is $b$ 
then stop irrespective of its value.

\section{Unrestricted Class $\Pi$: An Informal Discussion}
\label{general_class_section}
In this section, based on the insights we have obtained from the analysis in the previous sections,
we will informally discuss the possible structure of the optimal policy 
within the unrestricted class of policies, $\Pi$. 

Recall that a policy within $\Pi$, at stage $k$,
is in general allowed to base its decision on ($b_k,\mathcal{F}_k)$
where $b_k$ is the reward of the best probed relay ($b_k=-\infty$ if no relay has been
probed yet) and $\mathcal{F}_k$ is the set of
unprobed relays ($\mathcal{F}_k=\{\}$ if all the relays have been probed). 
Thus, the state space at stage $k$ can be written as
\begin{eqnarray}
\label{state_space_equn}
\mathcal{X}_k 
= \Big\{(b,\mathcal{H}): b\in\{-\infty\}\cup[0,\overline{r}], 
\mathcal{H}\in\mathcal{F}^j, 0\le j \le k \Big\}.
\end{eqnarray}
Again the actions available are {stop}, {probe}, and {continue}. If the action 
is to {probe} then $\mathscr{F}$ has to further decide which relay to probe among 
the several ones available at stage $k$.
When there are no unprobed relays
(i.e., $\mathcal{H}=\{\ \}$) we will represent the state as simply $b$.
We now proceed to write the recursive Bellman optimality 
equation for this more general unrestricted problem.
Although these equations are more involved than the ones in Section~\ref{restricted_class_section} 
(recall (\ref{cost_to_go_stageN_statebF_equn}) through (\ref{cost_to_go_stagek_statebF_equn})),
these can be understood similarly and hence we do not provide an explanation. The sole purpose
for writing these equations here is because we will require 
these (in Section~\ref{probing_numerical_work_section}) 
to perform value iteration and  numerically compute an optimal policy for the unrestricted problem. 
Hence these equations can be omitted
without affecting the readability of the remainder of this section.

Let $J_k$, $k=1,2,\cdots,N$, represent
the optimal cost-to-go at stage $k$ (for simplicity we are again using $J_k$), then,
$J_N(b) = -\eta b$, and $J_N(b,\mathcal{H})$ is as in (\ref{bellman_gen_last_equn}).
For stage $k=1,2,\cdots,N-1$ we have
\begin{eqnarray}
\label{bellman_gen_k_b_equn}
J_k(b) 
&=& \min\Big\{-\eta b, \tau + \mathbb{E}_L\Big[J_k(b,\{F_{L_{k+1}}\})\Big]\Big\},
\end{eqnarray}
and $J_k(b,\mathcal{H})$ as in (\ref{bellman_gen_k_G_equn})

In view of the complexity of the problem, we do not pursue the formal 
analysis of characterizing the structure of the optimal policy 
within the unrestricted
class. However, based on our results from the previous sections and 
a related work by Chaporkar and Proutiere \cite{chaporkar-proutiere08joint-probing},
we will discuss the possible structure of the unrestricted-optimal policy.

\subsection{Discussion on the Last Stage $N$}
Suppose the decision process enters the last stage $N$. 
Now, given the best reward value among the probed relays, $b$, 
and the set $\mathcal{H}$ of reward distributions of the unprobed
relays, $\mathscr{F}$ has to decide whether to {stop}, or {probe} a relay (note that 
{continue} action is not available at the last stage). Suppose the action
is to probe then, after probing and updating the best reward value, 
if still there are some unprobed relays left, $\mathscr{F}$ has to again decide
to stop or probe. This decision problem 
is similar to the one studied by Chaporkar and Proutiere 
in \cite{chaporkar-proutiere08joint-probing}, but from the context of channel 
selection. In the following, we will briefly describe the 
problem in \cite{chaporkar-proutiere08joint-probing}.

Given a set of channels with 
different channel gain distributions, a transmitter has to choose a channel
for its transmissions. The transmitter can probe a channel to know its channel gain.
Probing all the channels will
enable the transmitter to select the best channel but at the cost of reducing the
effective transmission time within the channel coherence period. On the other hand,
probing only a few channels may deprive the transmitter of the opportunity 
to transmit on a better channel. The transmitter is interested in \emph{maximizing}
its \emph{throughput} within the coherence period.


The authors in \cite{chaporkar-proutiere08joint-probing}, for their channel probing problem, 
prove that the one-step-look-ahead (OSLA)
rule is optimal: given the channel gain of the best channel (among the channels probed so far)
and a collection of channel gain distributions of the unprobed channels, 
it is optimal to stop and transmit on the best channel 
if and only if the throughput obtained by doing so is greater than the expected throughput obtained 
by probing any unprobed channel and then stopping (by transmitting on the new-best channel). 
Further, they prove that if the set of channel gain distributions is totally stochastically ordered 
(recall Definition~\ref{total_stochastic_ordering_definition}), then it is optimal to probe 
the channel whose distribution is stochastically largest among all the unprobed channels.
However, in their problem maximizing throughput involves optimizing a product of the channel gain and
the remaining transmission time, unlike in our problem where (at the last stage) we optimize
a linear combination of reward and the probing cost. But, from our numerical work we 
have seen that a similar OSLA rule is optimal once 
our decision process enters the last stage $N$: given a state $(b,\mathcal{H})$
at stage $N$, it is optimal to {stop} if the cost of stopping is less than the cost of probing 
any distribution from $\mathcal{H}$ and then stopping; otherwise it is optimal to 
{probe} the stochastically largest distribution from $\mathcal{H}$.

\subsection{Discussion on Stages $k=1,2,\cdots,N-1$}
For the other stages $k=1,2,\cdots,N-1$, one can begin by defining the stopping sets
$\mathcal{S}_k$ and $\mathcal{S}_k^\mathcal{H}$, and the stobing
sets $\mathcal{Q}_k^\mathcal{H}$, analogous to the ones in 
(\ref{optimal_stopping_equn}), (\ref{optimal_stopping_l_equn}) and
(\ref{optimal_stopping_probing_equn}).
Note that, here we need to define  $\mathcal{S}_k^\mathcal{H}$ and  $\mathcal{Q}_k^\mathcal{H}$
for a set of distributions $\mathcal{H}$ unlike in the earlier case where we had defined these
sets only for a given distribution $F_\ell$. We expect that the results 
analogous to the ones in Section~\ref{structural_results_section}, namely
Theorems~\ref{stopping_sets_equal_theorem} and \ref{stopping_sets_l_equal_theorem} where we prove 
that the stopping sets are stage independent, hold true for this more general setting as well.
Further, similar to that at stage $N$, for any stage $k$ we expect that if it is optimal to probe 
at some state $(b,\mathcal{H})$ then it is better to {probe} the stochastically 
largest distribution from $\mathcal{H}$.
Again, we have seen that these observations hold
in our numerical work.

\section{Numerical and Simulation Results}
\label{probing_numerical_work_section}

\subsection{One-Hop Study}

We begin by listing the various parameter values that we have used in our numerical work.
The forwarder and the sink are separated by a distance of $V=1000$ meters (m);
recall Fig.~\ref{forwarding_set_figure}. 
The radius of the communication region is $50$ m. We set
$z_{min}=5$ m. There are $N=5$ relays within the 
forwarding region $\mathcal{L}$. These are uniformly located within $\mathcal{L}$.
To enable us to perform value iteration (i.e., recursively solve the Bellman equation to 
obtain optimal value and the optimal policy), we have discretized the forwarding region $\mathcal{L}$
into a grid of $20$ uniformly spaced points within $\mathcal{L}$ and then 
map the location of each relay to a grid point closest to it.
Since the grid is symmetric about the line joining
$\mathscr{F}$ and the sink (with $4$ points lying on the line so that 
these do not have symmetric pairs), we have in total ($\frac{20-4}{2}+4=$) $12$ different possible 
$D_\ell$ values, giving
rise to $12$ different reward distributions constituting the set 
$\mathcal{F}$.

Next, recall the reward expression from (\ref{reward_equn});
we have fixed, $d_{ref}=5$ m,
$\xi = 2.5$, and $a = 0.5$. 
For $\Gamma N_0$, which is referred to as the \emph{receiver sensitivity}, we use 
a value of $10^{-9}$ mW (equivalently $-90$ dBm) specified for
the Crossbow TelosB wireless mote \cite{telosb-datasheet}. To ensure that the transmit power of 
a relay from any grid location is within the range of $1$ mW to $0.003$ mW 
(equivalently $0$ dBm to $-24$ dBm; again from TelosB datasheet \cite{telosb-datasheet}),\footnote{Although
practically only a finite set of transmit power levels will be allowed, 
for our numerical work we assume that the relays can transmit using any power within the specified
range.}
we allow for four different channel gain values:
$0.4\times 10^{-3}, 0.6\times 10^{-3}, 0.8\times  10^{-3}$, and  $1\times 10^{-3}$,
each occurring with equal probability.
Since channel probing is usually performed using the maximum allowable transmit power,
we set the probing cost $\delta$ to be $1$ mW.
Finally, the inter-wake-up times $\{U_k\}$ are exponentially
distributed random variables with mean $\tau = 20$ milliseconds (ms).

\textbf{\emph{One-Hop Policies:}} The following is the description of the policies that 
we will study:
\begin{itemize}
\item RST-OPT (ReSTricted OPTimal): The optimal policy within the restricted class 
(Sections~\ref{restricted_class_section} and \ref{structural_results_section})
where $\mathscr{F}$ is allowed to keep at most two relays awake $-$ the best probed
and the best unprobed relay; recall the implementation summary of this policy from 
Section~\ref{policy_implementation_section}. 

\item GLB-OPT (GLoBal OPTimal): The optimal policy within the unrestricted class
of policies where $\mathscr{F}$ operates by keeping all the unprobed relays
awake. We obtain GLB-OPT by numerically solving the optimality 
equations in (\ref{bellman_gen_last_equn}), (\ref{bellman_gen_k_G_equn}) and (\ref{bellman_gen_k_b_equn}).

\item BAS-OPT (BASic OPTtimal): The optimal policy for the basic relay selection model where $\mathscr{F}$
is not allowed to exercise the option of \emph{not-probing} a relay (recall discussion of the basic model from
related work). Thus, each time a relay wakes up,
it is immediately probed (incurring a cost of $\eta\delta$) and its reward value is revealed to 
$\mathscr{F}$. By incorporating $\eta\delta$ into the term $\tau$ (so that the 
inter-wake-up time is modified to $\tau+\eta\delta$), the solution to this model
can be characterized (see our prior work \cite[Section~6]{naveen-kumar12relay-selection_TMC_paper}) in terms of a single threshold $\alpha$
as follows: at any stage $k=1,2,\cdots,N-1$, stop if and only if the best reward value 
$b_k\ge\alpha$; at stage $N$ stop for any $b_N$. Note that
the threshold $\alpha$ depends on $\eta$.
\end{itemize}

\begin{figure}[t]
\centering
\includegraphics[scale=0.4]{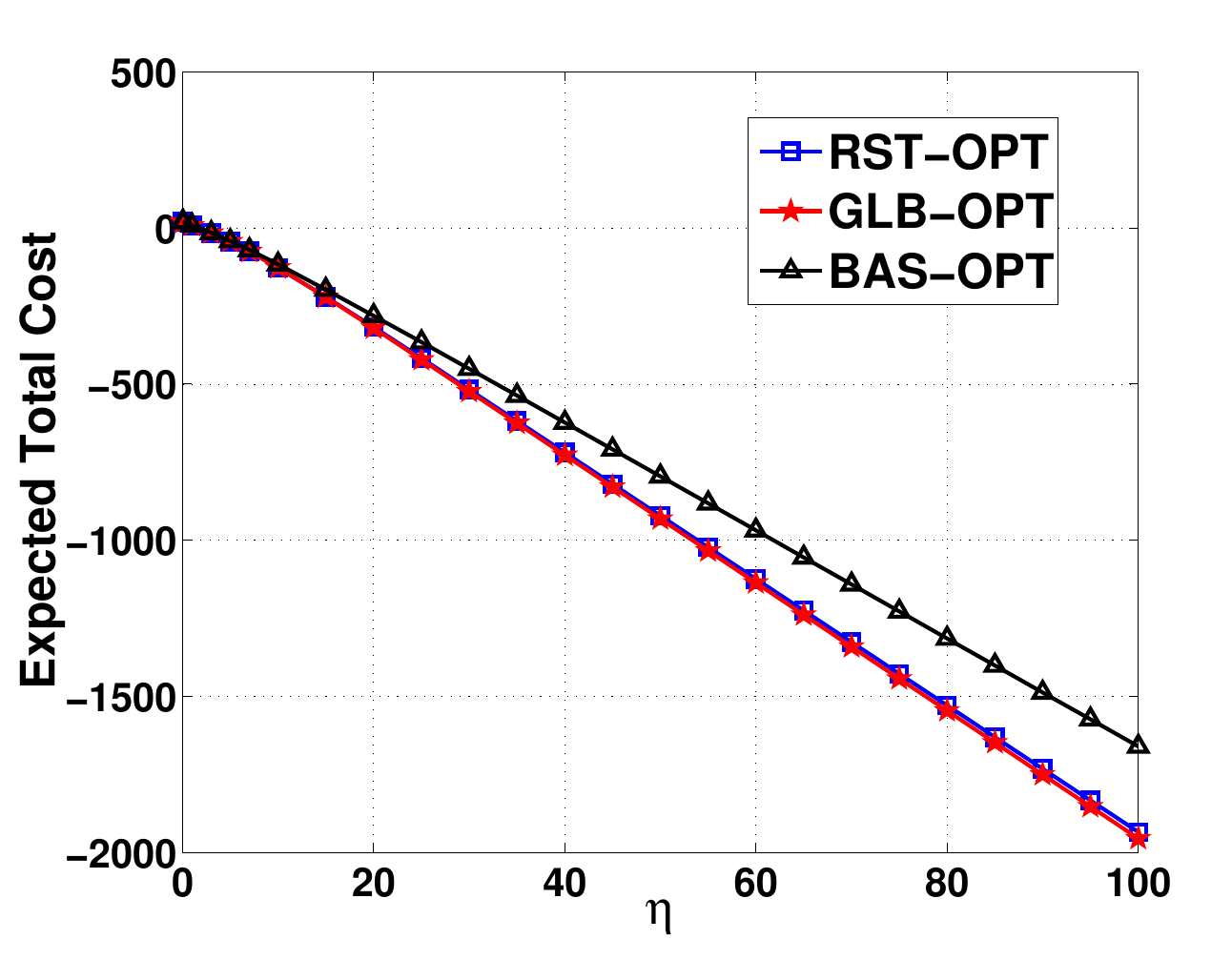}
\vspace{-4mm}
\caption{\label{total_cost_figure} Expected total cost 
as a function of the trade-off multiplier $\eta$; see (\ref{unconstrained_equn}). Recall that a large 
$\eta$ implies less emphasis on expected delay.}
\vspace{-5mm}
\end{figure}

\begin{figure*}
\centering
\subfigure[]{
\includegraphics[scale=0.31]{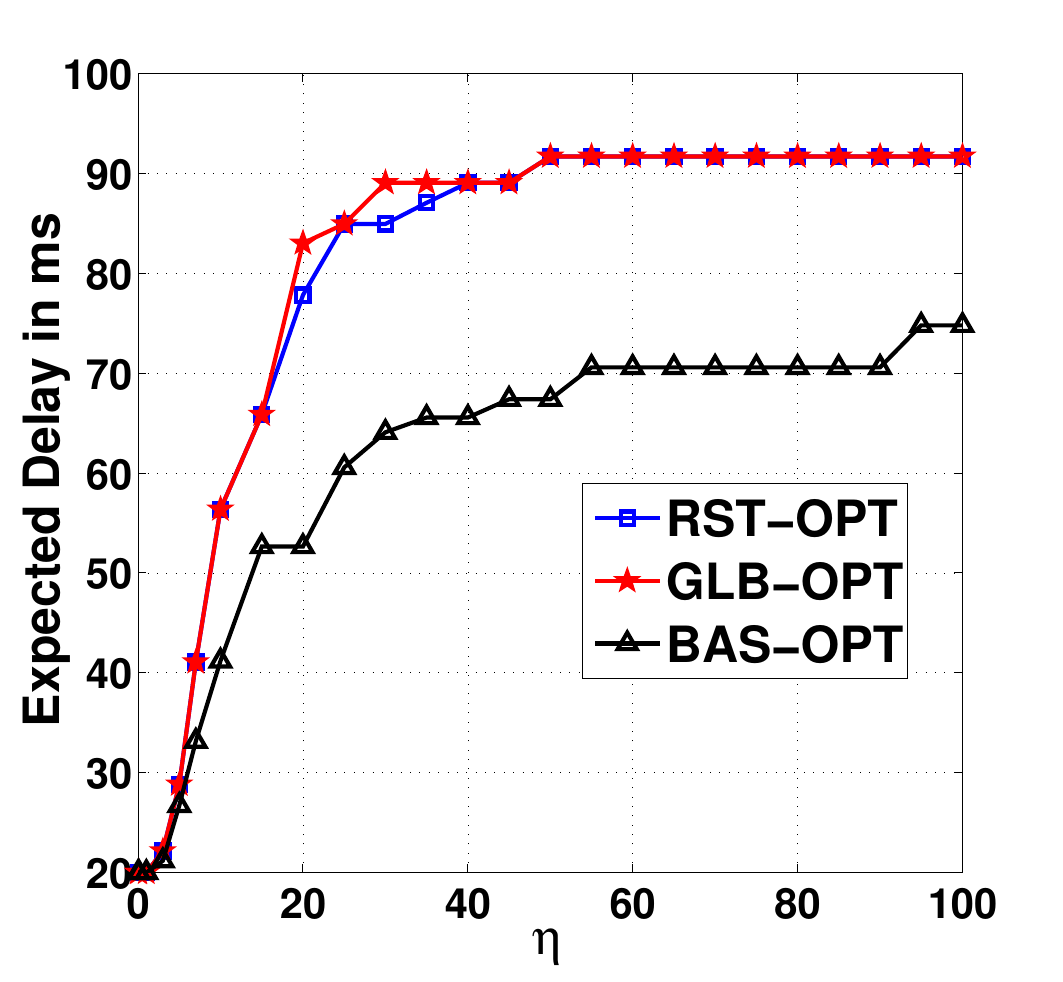}
        \label{delay_figure}
}
\subfigure[]{
\includegraphics[scale=0.31]{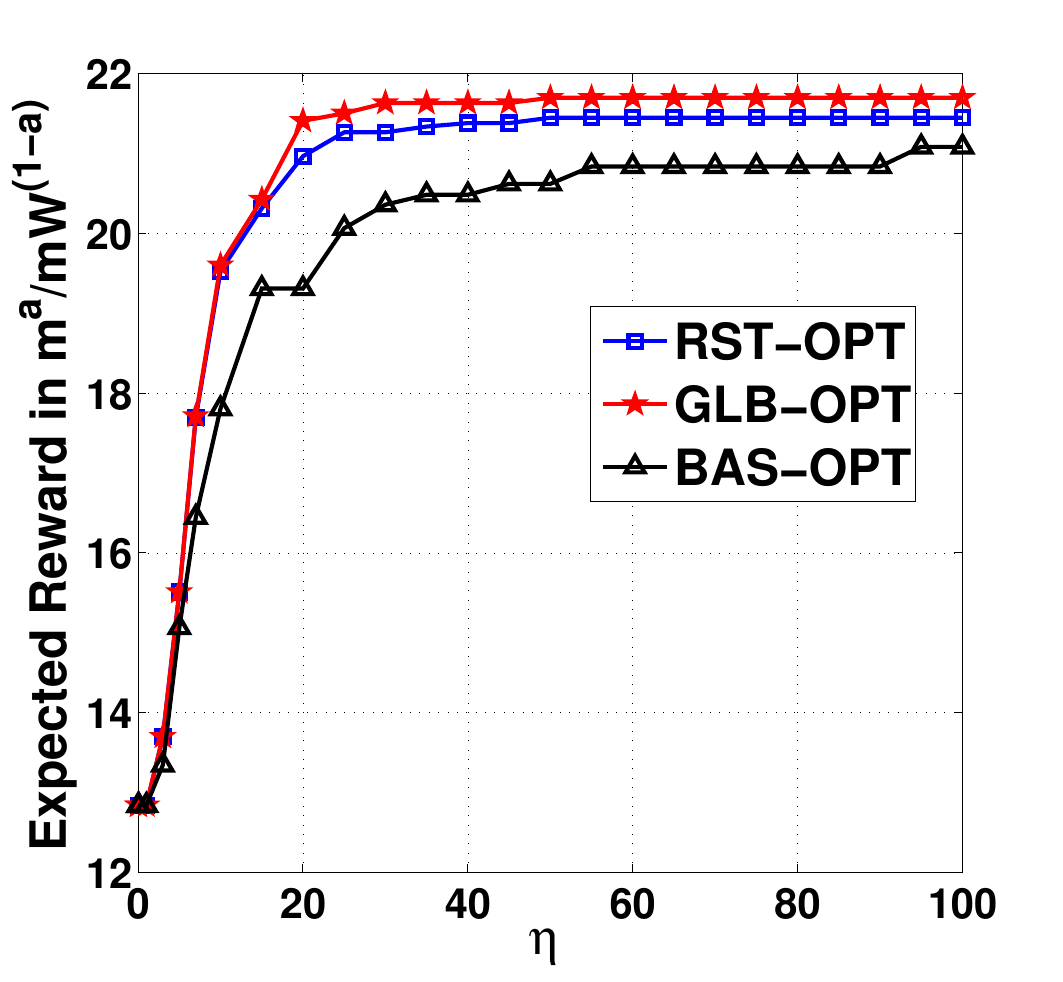}
        \label{reward_figure} 
}
\subfigure[]{
\includegraphics[scale=0.31]{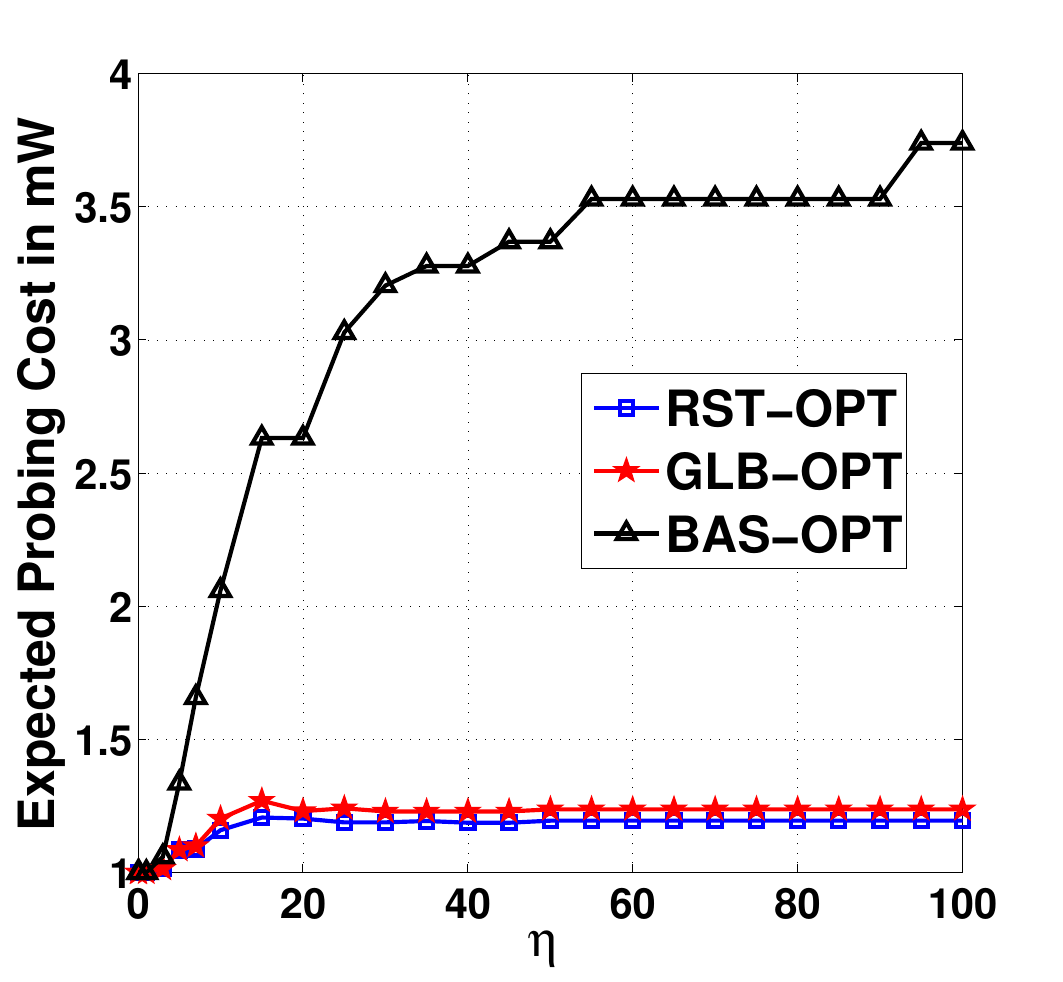}
\label{probing_figure}
}
\caption{\label{individual_components_figure}
Individual components of the total cost in Fig.~\ref{total_cost_figure} as functions of $\eta$:
\subref{delay_figure} Delay \subref{reward_figure} Reward  
and \subref{probing_figure} Probing Cost.
}
\vspace{-4mm}
\end{figure*}

\textbf{\emph{Discussion:}}
In Fig.~\ref{total_cost_figure} we have plotted the total cost 
(i.e., the objective in (\ref{unconstrained_equn})) incurred by 
each of the above policies as a function of the multiplier $\eta$.
GLB-OPT being the globally optimal policy achieves the minimum cost.
However, interestingly we observe that the total cost obtained by RST-OPT
is very close to that of GLB-OPT. While the performance of BAS-OPT is good for 
small values of $\eta$, the performance degrades
as $\eta$ increases illustrating that it is not wise to naively probe every relay
as and when they wake-up.



In Fig.~\ref{individual_components_figure} we have shown
the individual components 
of the total cost (namely delay, reward, and probing cost) 
as functions of $\eta$.
As $\eta$ decreases to $0$ we see (from Fig.~\ref{delay_figure})
that the expected delay incurred by all the policies converges to $20$ ms which is the mean 
time, $\tau$, until the first relay wakes up. Similarly, the expected rewards (in Fig.~\ref{reward_figure})  converge to 
reward of the first relay, and the probing costs (in Fig.~\ref{probing_figure}) converge
to the cost of probing a single relay, i.e., $\delta=1$ mW. This is because, 
for small values of $\eta$, since delay is valued more (recall the total cost expression from (\ref{unconstrained_equn})),
all the policies essentially end up probing the first relay and then forwarding the packet to it. 
This also explains as to why similar total cost (recall Fig.~\ref{total_cost_figure}) 
is incurred by all the policies in the low $\eta$ regime (e.g., $\eta\le 20$).

Next, as $\eta$ increases we see that the delay incurred and the reward achieved by all the policies
increases (see Fig.~\ref{delay_figure} and \ref{reward_figure}, respectively).
While the probing cost of BAS-OPT naively increases (see Fig.~\ref{probing_figure}),
probing costs incurred by RST-OPT and GLB-OPT saturate beyond $\eta=20$. This is because, whenever $\eta$ is large,
RST-OPT and GLB-OPT are aware that the gain in reward value obtained by probing more relays is negated by 
the cost term, $\eta\delta$, which is added to the total cost each time a new relay is probed;
BAS-OPT, not allowed to not-probe, ends up probing all the relays until the best reward exceeds 
the threshold $\alpha$. Thus, although BAS-OPT incurs a smaller delay than the other two policies, 
but suffers both in terms of reward and probing cost, leading to an higher total cost. 
On the other hand, RST-OPT and GLB-OPT wait for more relays and then probe only the relays with good reward
distribution to accrue a better total cost. 

Finally, the marginal improvement in performance obtained by GLB-OPT over RST-OPT can be understood as
follows. Although the delay incurred by these two policies is almost identical, 
for large $\eta$ values, GLB-OPT achieves a better reward than RST-OPT by incurring a slightly 
higher probing cost. Thus, whenever the reward offered by the relay with the best distribution
is not good enough, GLB-OPT probes an additional relay to improve the reward; such improvement 
is not possible by RST-OPT since it is restricted to keep only one unprobed relay awake. 
 
%

\textbf{\emph{Computational Complexity:}}
Finally on the computational complexity of these policies. To obtain GLB-OPT
we had to recursively solve the Bellman equation (referred to as the \emph{value iteration})
in  (\ref{bellman_gen_last_equn}), (\ref{bellman_gen_k_G_equn}) and (\ref{bellman_gen_k_b_equn}),
for every stage $k$ and every possible state at stage $k$.
The total number of all possible states at stage $k$, i.e., 
the cardinality of the state space $\mathcal{X}_k$ in (\ref{state_space_equn}),
grows {exponentially} with the cardinality of $\mathcal{F}$ (assuming that $\mathcal{F}$
is discrete like in our numerical example). It also grows {exponentially} with the stage index $k$.

In contrast, for computing RST-OPT, since within the restricted class at any time only one 
unprobed relay is kept awake, the state space size grows only 
{linearly} with the cardinality of $\mathcal{F}$.
Also, the size of the state space does not grow with $k$. Furthermore, from our analysis in 
Section~\ref{structural_results_section} we know that the stopping sets are threshold based,
and moreover the thresholds, $\alpha_k$ and $\{\alpha_k^\ell:F_\ell\in\mathcal{F}\}$,
are stage independent. 
Hence, these thresholds have to be computed only once (for stage $N-1$ and $N$, respectively), thus
further reducing the complexity of RST-OPT. BAS-OPT, being a single-threshold based policy,
is much simpler to implement but is not a good choice whenever $\eta$ is large.

\subsection{End-to-End Study}
The good one-hop performance of RST-OPT and its computational simplicity motivates
 us to apply RST-OPT to route packets in an asynchronously
sleep-wake cycling WSN and study its end-to-end performance. We will also obtain
the end-to-end performance of the naive BAS-OPT policy.

First we will describe the setting that we have considered for our end-to-end simulation study.
We construct a network by randomly placing $500$ nodes in
a square region of side $500$~m. The sink node is placed at the location $(500,0)$.
The network nodes are asynchronously and periodically sleep-wake cycling,
i.e., a node $i$ wakes up at the periodic instants, $\{T_i+kT:k\ge0\}$, where $\{T_i\}$ are 
i.i.d.\ uniform on $[0,T]$ with $T$ being the sleep-wake cycling period (recall
our justification for the periodic sleep-wake cycling from footnote~\ref{sleep-wake-footnote} in 
page~\pageref{sleep-wake-footnote}).
We fix $T=100$ ms. A source node is randomly chosen, which generates an alarm packet
at time $0$. This alarm packet has to be routed to the sink node.

\begin{figure*}[t]
\centering
\subfigure[]{
\includegraphics[scale=0.4]{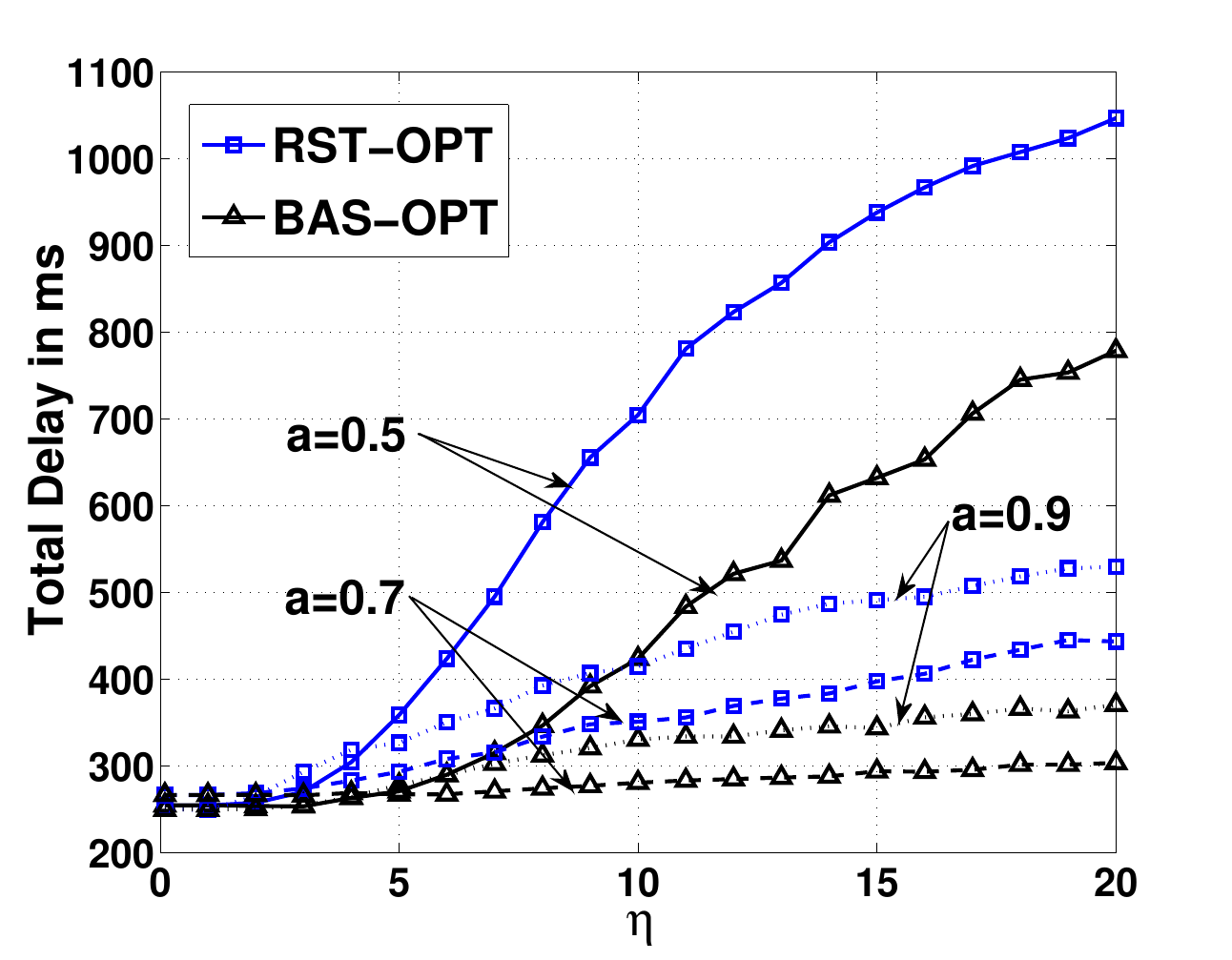}
        \label{end_delay_figure}
}
\subfigure[]{
\includegraphics[scale=0.4]{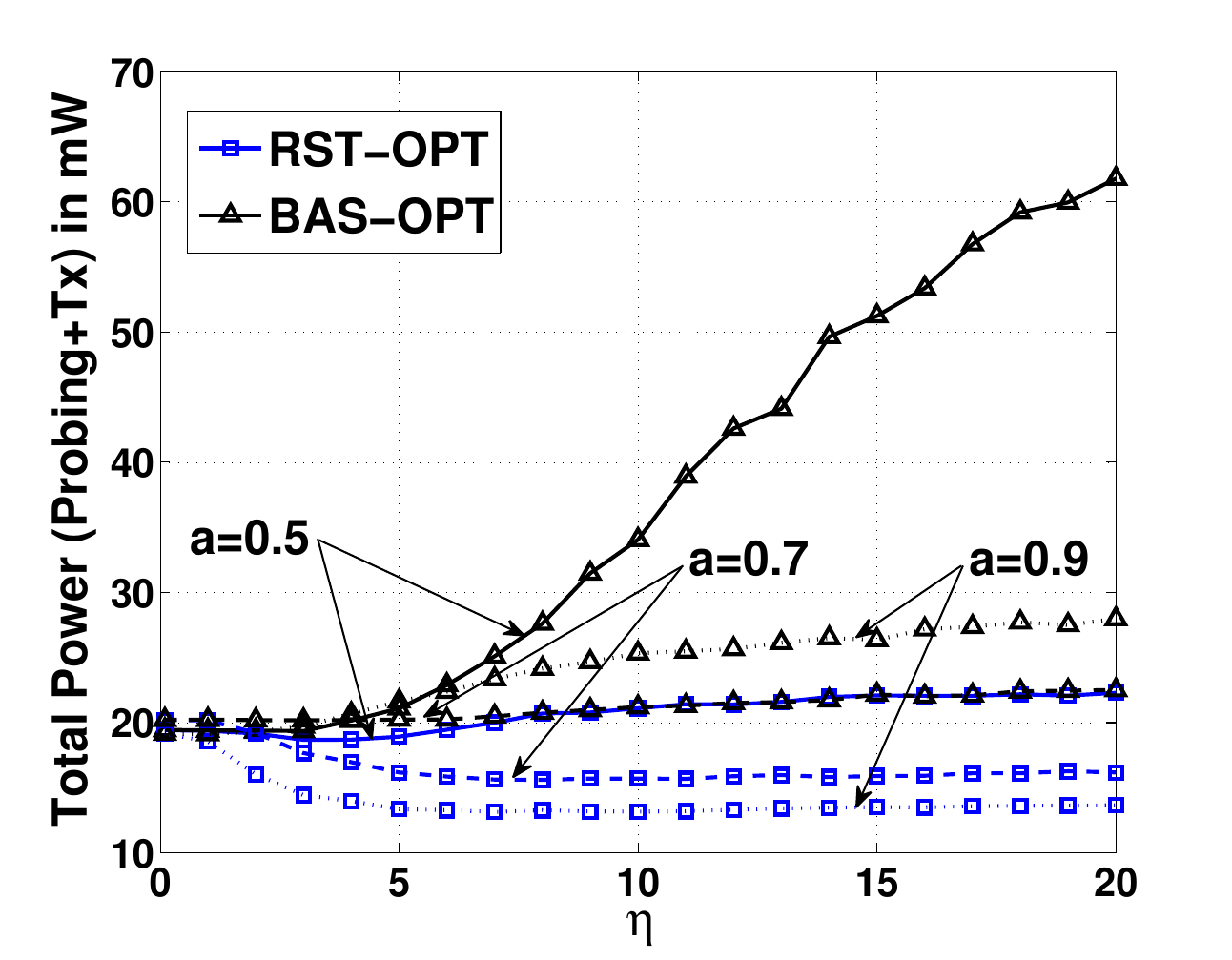}
        \label{end_power_figure} 
}
\vspace{-3mm}
\caption{\label{end_figure} End-to-end performance of RST-OPT and BAS-OPT as functions of $\eta$ for different values of $a$: 
\subref{end_delay_figure} Total delay, and \subref{end_power_figure} Total power.}
\vspace{-4mm}
\end{figure*}

Here, in addition to varying $\eta$, we will also vary the multiplier $a$ 
and study the end-to-end performance. Recall from  (\ref{reward_equn}) that 
$a$ is the multiplier used to trade-off between
progress and power in the reward expression; a larger value of $a$ implies more emphasis on progress. 
The values of all the other parameters, e.g., $r_c$, $\delta$, $\Gamma N_0$, channel gains, etc.,
remain as in our one-hop study.

Now, for a given $\eta$ and $a$,
each node computes the corresponding RST-OPT and BAS-OPT policies assuming a mean 
inter-wake-up time of $\frac{T}{N_i}$ ms,
where $N_i$ is the number of nodes in the forwarding region of node $i$.
In Fig.~\ref{end_figure}, for three different values of $a$ (namely $0.5$, $0.7$, and $0.9$) 
we have plotted, as functions of $\eta$,  the
total delay and the total power (which is the sum of the probing and the transmission powers
incurred at each hop) incurred,  by applying RST-OPT and BAS-OPT policies
at each hop en-route to the sink node. Each data point in Fig.~\ref{end_figure}
is obtained by averaging the respective quantities over $1000$ alarm packets. 

\emph{\textbf{Discussion:}} First, note that both total delay and 
total power incurred by BAS-OPT are increasing with 
$\eta$ for each $a$.  Hence, no favorable trade-off between delay and
power can be obtained using BAS-OPT; it is better to operate BAS-OPT
at a low value of $\eta$, where the total delay incurred is (approximately) 250~ms
while the total power expended is about 20~mW. 
In fact, as $\eta$ decreases to $0$, we see that the performance of all the policies
(i.e., RST-OPT and BAS-OPT for different values of $a$)
converge to these values. This is simply because, whenever $\eta$ is small, since 
(one-hop) delay is valued more, all the policies, at each hop, essentially forward the packet
to the first relay that wakes up.

For RST-OPT, while only a marginal trade-off 
between delay and power can be achieved for $a=0.5$ (see from Fig.~\ref{end_power_figure}
that the corresponding total power decreases only marginally as $\eta$ increases from $1$ to $4$),
but as we increase the value of $a$ to $0.7$ and then to $0.9$,  we see that the total power sharply
decreases with $\eta$. For instance, for $a=0.9$, from Fig.~\ref{end_power_figure} we see that the
total power decrease from $20$~mW to $13$ mW as $\eta$ goes from $0$ to $7$. However, over this range
of $\eta$, total delay increases from $250$ ms to $360$ ms (see the plot corresponding to RST-OPT, $a=0.9$,
from Fig.~\ref{end_delay_figure}). Thus, for these higher values of $a$, trade-off between delay and power can be achieved using
RST-OPT.


Next, for any fixed $\eta$, from Fig.~\ref{end_power_figure} observe that the total power
incurred by RST-OPT is improving (i.e., decreasing) with $a$. This can be understood 
as follows: since a larger $a$ gives less emphasis on power and more 
emphasis on progress in the reward expression (recall (\ref{reward_equn})),
then, although the one-hop transmissions may be of higher power, but
there are fewer hops and hence fewer transmissions, thus resulting in a lower total power.
This observation would suggest that it is advantageous to use RST-OPT by setting
$a=0.9$ rather than $a=0.5$ or $0.7$.
However, from Fig.~\ref{end_delay_figure} we see that the total delay is not decreasing with $a$.
In fact, delay incurred by RST-OPT first decreases as $a$ increases from $0.5$ to $0.7$,
and then increases as $a$ is further increased to $0.9$. Similar is the case for
the plots corresponding to BAS-OPT in Fig.~\ref{end_delay_figure}.
This observation can be understood as follows. When $a=0.5$, since (one-hop) power is valued more, 
the respective forwarding nodes at each hop 
will end up spending more time waiting for a relay which require strictly lesser transmission power. 
Similarly, when $a=0.9$, larger delay is incurred at each hop since
the forwarding nodes now have to wait for 
relays whose progress value is more (however, since $a=0.9$ results in a fewer hops we see that the delay 
incurred in this case is considerably less than the $a=0.5$ case). On the other hand,  
when $a=0.7$, since a relatively fair trade-off between progress and power exists, 
the waiting time at each hop is 
reduced because now any relay with a moderate progress and  a moderate transmission power would suffice.

 
The above argument is precisely the reason as to why the total power incurred by BAS-OPT 
behaves as in Fig.~\ref{end_power_figure}:
when $a=0.5$ or $0.9$, each forwarder, in the process of waiting 
for a relay whose transmission power requirement is low or progress is large, respectively, 
will end up probing more relays. RST-OPT benefits over BAS-OPT
here by  probing  only good relays at each hop, thus yielding a lower total power. 

Finally, summarizing our end-to-end results, we see that no trade-off between 
delay and power can be achieved by the naive BAS-OPT policy, while RST-OPT 
achieves such a trade-off (by varying $\eta$) for $a=0.7$ or $0.9$. Further, 
for a fixed $\eta$, favorable trade-off between delay and power can be obtained 
by varying $a$. For instance, from Fig.~\ref{end_figure} we see that when $\eta=7$,
moving from $a=0.7$ to $0.9$ will result in a power saving of about $3$ mW
while increasing the end-to-end delay by $130$ ms. Thus, depending on the application requirement 
(i.e., delay or power sensitive application) one has to appropriately choose the values of 
$\eta$ and $a$.

\section{Conclusion}
\label{prb_conclusion_section}
Motivated by the problem of end-to-end geographical forwarding in a sleep-wake cycling wireless sensor network,
we formulated a decision problem of choosing a next-hop relay node when a set of potential relay
neighbors are sequentially waking up in time. A power cost is incurred for probing a relay to
learn its channel gain.  
We first studied a restricted class of policies where a policy's decision is 
based only on, in addition to the best probed relay, the best 
unprobed relays (instead of all the unprobed 
relays). We characterized the optimal policy in terms of stopping sets. 
Our first main result (Theorem~\ref{threshold_nature_lemma}) was to show that the 
stopping sets are threshold based. Then we proved that the stopping sets are stage independent
(Theorem~\ref{stopping_sets_equal_theorem} and \ref{stopping_sets_l_equal_theorem}).
A discussion on the more general unrestricted class of policies was provided.
We conducted numerical work to compare the performances of the restricted optimal (RST-OPT)
and the global optimal (GLB-OPT) policies. We observed that the performance of
RST-OPT is close to that of GLB-OPT. We also conducted simulation experiments to 
study the end-to-end performance of RST-OPT. Finally, it is worth noting that our work being a variant of the
asset selling problem, can, in general, find application
wherever the problem of resource-selection occurs,
when a collection of resources are sequentially arriving.

\balance
\bibliographystyle{IEEEtran}
\bibliography{IEEEabrv,naveen-kumar12relay-selection-with-channel-probing}

\begin{thebibliography}{10}
\providecommand{\url}[1]{#1}
\csname url@samestyle\endcsname
\providecommand{\newblock}{\relax}
\providecommand{\bibinfo}[2]{#2}
\providecommand{\BIBentrySTDinterwordspacing}{\spaceskip=0pt\relax}
\providecommand{\BIBentryALTinterwordstretchfactor}{4}
\providecommand{\BIBentryALTinterwordspacing}{\spaceskip=\fontdimen2\font plus
\BIBentryALTinterwordstretchfactor\fontdimen3\font minus
  \fontdimen4\font\relax}
\providecommand{\BIBforeignlanguage}[2]{{%
\expandafter\ifx\csname l@#1\endcsname\relax
\typeout{** WARNING: IEEEtran.bst: No hyphenation pattern has been}%
\typeout{** loaded for the language `#1'. Using the pattern for}%
\typeout{** the default language instead.}%
\else
\language=\csname l@#1\endcsname
\fi
#2}}
\providecommand{\BIBdecl}{\relax}
\BIBdecl

\bibitem{kim-etal09optimal-anycast}
J.~Kim, X.~Lin, and N.~Shroff, ``{Optimal Anycast Technique for Delay-Sensitive
  Energy-Constrained Asynchronous Sensor Networks},'' \emph{IEEE/ACM
  Transactions on Networking}, April 2011.

\bibitem{naveen-kumar12relay-selection_TMC_paper}
K.~P. Naveen and A.~Kumar, ``{Relay Selection for Geographical Forwarding in
  Sleep-Wake Cycling Wireless Sensor Networks},'' \emph{IEEE Transactions on
  Mobile Computing}, vol.~12, no.~3, pp. 475--488, 2013.

\bibitem{bertsekas-tsitsiklis91stochastic-shortest-path}
D.~P. Bertsekas and J.~N. Tsitsiklis, ``An {A}nalysis of {S}tochastic
  {S}hortest {P}ath {P}roblems,'' \emph{Mathematics of Operations Research},
  vol.~16, 1991.

\bibitem{naveen-kumar10geographical-forwarding}
K.~P. Naveen and A.~Kumar, ``Tunable {L}ocally-{O}ptimal {G}eographical
  {F}orwarding in {W}ireless {S}ensor {N}etworks with {S}leep-{W}ake {C}ycling
  {N}odes,'' in \emph{INFOCOM 2010, 29th IEEE Conference on Computer
  Communications}, March 2010.

\bibitem{rappaport01wireless-communication}
T.~Rappaport, \emph{{Wireless Communications: Principles and Practice}},
  2nd~ed.\hskip 1em plus 0.5em minus 0.4em\relax Upper Saddle River, NJ, USA:
  Prentice Hall PTR, 2001.

\bibitem{thejaswi-etal10two-level-probing}
P.~S.~C. Thejaswi, J.~Zhang, M.~O. Pun, H.~V. Poor, and D.~Zheng, ``Distributed
  {O}pportunistic {S}cheduling with {T}wo-{L}evel {P}robing,'' \emph{IEEE/ACM
  Transactions on Networking}, vol.~18, no.~5, October 2010.

\bibitem{bertsekas05optimal-control-vol1}
D.~P. Bertsekas, \emph{Dynamic Programming and Optimal Control, Vol. I}.\hskip
  1em plus 0.5em minus 0.4em\relax Athena Scientific, 2005.

\bibitem{karlin62selling-asset}
S.~Karlin, \emph{Stochastic Models and Optimal Policy for Selling an
  Asset}.\hskip 1em plus 0.5em minus 0.4em\relax Stanford University Press,
  Stanford, 1962.

\bibitem{akkaya-younis05survey}
K.~Akkaya and M.~Younis, ``A {S}urvey on {R}outing {P}rotocols for {W}ireless
  {S}ensor {N}etworks,'' \emph{Ad Hoc Networks}, vol.~3, pp. 325--349, 2005.

\bibitem{zorzi-rao03geographicrandom}
M.~Zorzi and R.~R. Rao, ``Geographic {R}andom {F}orwarding ({G}e{R}a{F}) for
  {A}d {H}oc and {S}ensor {N}etworks: {M}ultihop {P}erformance,'' \emph{IEEE
  Transactions on Mobile Computing}, vol.~2, pp. 337--348, 2003.

\bibitem{cinlar75stochastic-processes}
E.~Cinlar, \emph{{Introduction to Stochastic Processes}}.\hskip 1em plus 0.5em
  minus 0.4em\relax Prentice-Hall, 1975.

\bibitem{kumar-etal08wireless-networking}
A.~Kumar, D.~Manjunath, and J.~Kuri, \emph{Wireless Networking}.\hskip 1em plus
  0.5em minus 0.4em\relax San Francisco, CA, USA: Morgan Kaufmann Publishers
  Inc., 2008.

\bibitem{agrawal-patwari09correlated-shadow_journal}
P.~Agrawal and N.~Patwari, ``{Correlated Link Shadow Fading in Multi-Hop
  Wireless Networks},'' \emph{IEEE Transactions on Wireless Communications},
  vol.~8, no.~8, pp. 4024--4036, 2009.

\bibitem{albright77generalized-house-selling}
S.~C. Albright, ``{A Bayesian Approach to a Generalized House Selling
  Problem},'' \emph{Management Science}, vol.~24, no.~4, pp. 432--440, 1977.

\bibitem{rosenfield-etal83selling-asset}
D.~B. Rosenfield, R.~D. Shapiro, and D.~A. Butler, ``{Optimal Strategies for
  Selling an Asset},'' \emph{Management Science}, vol.~29, no.~9, pp.
  1051--1061, 1983.

\bibitem{mauve-hartenstein01survey}
M.~Mauve, J.~Widmer, and H.~Hartenstein, ``A {S}urvey on {P}osition-{B}ased
  {R}outing in {M}obile {A}d-{H}oc {N}etworks,'' \emph{IEEE Network}, vol.~15,
  pp. 30--39, 2001.

\bibitem{liu-etal07CMAC}
S.~Liu, K.~W. Fan, and P.~Sinha, ``C{M}{A}{C}: An {E}nergy {E}fficient
  {M}{A}{C} {L}ayer {P}rotocol using {C}onvergent {P}acket {F}orwarding for
  {W}ireless {S}ensor {N}etworks,'' in \emph{SECON '07: 4th Annual IEEE
  Communications Society Conference on Sensor, Mesh and Ad Hoc Communications
  and Networks}, June 2007, pp. 11--20.

\bibitem{chaporkar-proutiere08joint-probing}
P.~Chaporkar and A.~Proutiere, ``Optimal {J}oint {P}robing and {T}ransmission
  {S}trategy for {M}aximizing {T}hroughput in {W}ireless {S}ystems,''
  \emph{IEEE Journal on Selected Areas in Communications}, vol.~26, no.~8, pp.
  1546--1555, October 2008.

\bibitem{chang-liu07channel-probing}
N.~B. Chang and M.~Liu, ``Optimal {C}hannel {P}robing and {T}ransmission
  {S}cheduling for {O}pportunistic {S}pectrum {A}ccess,'' in \emph{MobiCom '07:
  Proceedings of the 13th annual ACM international conference on Mobile
  computing and networking}, 2007, pp. 27--38.

\bibitem{stadje97two-levels}
W.~Stadje, ``An {O}ptimal {S}topping {P}roblem with {T}wo {L}evels of
  {I}ncomplete {I}nformation,'' \emph{Mathematical Methods of Operations
  Research}, vol.~45, pp. 119--131, 1997.

\bibitem{stoyan83comparison-methods-queues}
D.~Stoyan, \emph{Comparison Methods for Queues and other Stochastic
  Models}.\hskip 1em plus 0.5em minus 0.4em\relax John Wiley \& Sons, New York,
  1983.

\bibitem{telosb-datasheet}
\BIBentryALTinterwordspacing
Crossbow, ``{TelosB Mote Platform}.'' [Online]. Available:
  \url{www.willow.co.uk/TelosB_Datasheet.pdf}
\BIBentrySTDinterwordspacing

\end{thebibliography}

\onecolumn
\appendices

\section{Proof of Lemma~\ref{cost_to_go_ordering_lemma}}
\label{cost_to_go_ordering_lemma_appendix}
For convenience, here in the appendix, we will recall the respective Lemma/Theorem 
statements before providing their
proofs. Now, before proceeding to the proof of Lemma~\ref{cost_to_go_ordering_lemma},
we will require the following result first.

\begin{lemma}
\label{cost_decreasing_lemma}
For $k=1,2,\cdots,N$, $J_k(b)$ and $J_k(b,F_\ell)$ are decreasing in $b$.
\end{lemma}
\begin{IEEEproof}
Proof is by induction. For stage $N$ we know that
$J_N(b)=-\eta b$, and hence is decreasing in $b$. Also, recalling  $J_N(b,F_\ell)$ from  (\ref{cost_to_go_stageN_statebF_equn}):
\begin{eqnarray*}
J_N(b,F_\ell)=\min\Big\{-\eta b, \eta\delta-\eta\mathbb{E}_\ell\Big[\max\{b,R_\ell\}\Big]\Big\},
\end{eqnarray*}
it is easy to see that $J_N(b,F_\ell)$
is also decreasing in $b$. Thus, the monotonicity properties holds for stage $N$. 
Now, suppose $J_{k+1}(b)$ and $J_{k+1}(b,F_\ell)$ (for all $F_\ell$)
are decreasing in $b$ for some $k+1=2,3,\cdots,N$, then we will show that the 
result holds for stage $k$ as well.

First, recall the expressions of $J_k(b)$ and $J_k(b,F_\ell)$ 
(from (\ref{Jk_b_equn}) and (\ref{Jk_bF_equn}) respectively):
$J_k(b)=\min\Big\{-\eta b, C_k(b)\Big\}$ and
$J_k(b,F_\ell)=\min\Big\{-\eta b, P_k(b,F_\ell), C_k(b,F_\ell)\Big\}$.
Thus to complete the proof it is sufficient to show that $C_k(b)$, $P_k(b,F_\ell)$ and $C_k(b,F_\ell)$ 
are decreasing in $b$.
From the induction hypothesis, it is easy to see that $C_k(b)$ 
(in (\ref{continuing_cost_b_equn})) 
is decreasing in $b$, so that we obtain $J_k(b)$ is decreasing in $b$. 
Now that we have established $J_k(b)$ is decreasing in $b$,
it will immediately follow  that the 
probing cost $P_k(b,F_\ell)$ (in (\ref{probing_cost_bF_equn})) is decreasing in $b$.
Finally, again using the induction argument, observe that 
$\min\Big\{J_{k+1}(b,F_\ell),J_{k+1}(b,F_{L_{k+1}})\Big\}$
is decreasing in $b$ so that the continuing cost $C_k(b,F_\ell)$ 
(in (\ref{continuing_cost_bF_equn})) is also decreasing.
\end{IEEEproof}

\verb11

We are now ready to prove Lemma~\ref{cost_to_go_ordering_lemma}.

\emph{Lemma}~\ref{cost_to_go_ordering_lemma}\emph{:}
\begin{enumerate}
\item[(i)] For $k=1,2,\cdots,N-1$, if $F_\ell\ge_{st}F_u$ then $C_k(b,F_\ell)\le C_k(b,F_u)$, 
(including $k=N$) $P_k(b,F_\ell)\le P_k(b,F_u)$ and $J_{k}(b,F_\ell)\le J_{k}(b,F_u)$.
\item[(ii)] For $k=1,2,\cdots,N-2$, $C_k(b)\le C_{k+1}(b)$ and $C_k(b,F_\ell)\le C_{k+1}(b,F_\ell)$, 
(including $k=N-1$) $P_k(b,F_\ell)\le P_{k+1}(b,F_\ell)$ and $J_k(b,F_\ell)\le J_{k+1}(b,F_\ell)$.
\end{enumerate}

\begin{IEEEproof}[Proof of Part-(i)]
Consider stage $N$ and recall the expression for the optimal cost-to-go function $J_N(b,F_\ell)$ from (\ref{cost_to_go_stageN_statebF_equn}):
\begin{eqnarray*}
J_N(b,F_\ell) 
&=&\min\Big\{-\eta b, P_N(b,F_\ell)\Big\}\\
&=&\min\Big\{-\eta b, \eta\delta-\eta\mathbb{E}_\ell\Big[\max\{b,R_\ell\}\Big]\Big\}.
\end{eqnarray*}
Since the function $f(r)=\max\{b,r\}$ is increasing in $r$,
using the definition of stochastic ordering (Definition~\ref{stochastic_ordering_definition}) 
we can write 
\begin{eqnarray*}
\mathbb{E}_\ell\Big[\max\{b,R_\ell\}\Big]\ge\mathbb{E}_u\Big[\max\{b,R_u\}\Big],
\end{eqnarray*}
so that we have $P_N(b,F_\ell)\le P_N(b,F_u)$ and $J_N(b,F_\ell)\le J_N(b,F_u)$. Thus, the result holds for stage $N$.

Now suppose the result is true for some $k+1=2,3,\cdots,N$. 
From Lemma~\ref{cost_decreasing_lemma} we know that $J_k(b)$ is decreasing in $b$, which would
imply that, for any $b$, the function $f(r)=J_k(\max\{b,r\}))$ is decreasing in $r$.
Again, using the definition of  stochastic ordering (in 
Definition~\ref{stochastic_ordering_definition})
we can conclude that 
\begin{eqnarray*}
\mathbb{E}_\ell\Big[J_k(\max\{b,R_\ell\})\Big]\le\mathbb{E}_u\Big[J_k(\max\{b,R_u\})\Big],
\end{eqnarray*}
so that $P_k(b,F_\ell)\le P_k(b,F_u)$ (see (\ref{probing_cost_bF_equn})).
Next, from the induction argument we know that $J_{k+1}(b,F_\ell)\le J_{k+1}(b,F_u)$ so that
\begin{eqnarray*}
\min\Big\{J_{k+1}(b,F_\ell),J_{k+1}(b,F_{L_{k+1}})\Big\} 
\le \min\Big\{J_{k+1}(b,F_u),J_{k+1}(b,F_{L_{k+1}})\Big\}.
\end{eqnarray*}
Therefore, we also have $C_k(b,F_\ell)\le C_k(b,F_u)$ (see (\ref{continuing_cost_bF_equn})).
The proof can now be easily completed by recalling (from (\ref{Jk_bF_equn})) that 
$J_k(b,F_\ell)=\min\Big\{-\eta b, P_k(b,F_\ell), C_k(b,F_\ell)\Big\}$.\\

\emph{Proof of Part-(ii):}
This result is very intuitive, since with more number of stages 
to go, one is expected to accrue a lower cost. However, we prove it here 
for completeness. Again the proof is by induction. For stage $N-1$ we easily have, 
\begin{eqnarray*}
J_{N-1}(b)
&=&\min\Big\{-\eta b, C_k(b)\Big\}\\
&\le& -\eta b \\
&=& J_N(b).
\end{eqnarray*}
Next, consider a state of the form $(b,F_\ell)$.
The cost of probing $P_{N-1}(b,F_\ell)$ can be bounded as follows:
\begin{eqnarray*}
P_{N-1}(b,F_\ell)
&=&\eta\delta+\mathbb{E}_\ell\Big[J_{N-1}(\max\{b,R_\ell\})\Big]\\
&\overset{*}{\le}&\eta\delta+\mathbb{E}_\ell\Big[J_{N}(\max\{b,R_\ell\})\Big]\\
&\overset{o}{=}&\eta\delta-\eta\mathbb{E}_\ell\Big[\max\{b,R_\ell\}\Big]\\
&\overset{\dagger}{=}& P_N(b,F_\ell),
\end{eqnarray*}
where, to obtain $*$ we have used, $J_{N-1}(b)\le J_{N}(b)$ (which we had just proved),
$o$ is because $J_N(b)=-\eta b$ for all $b$, and  $\dagger$ is simply obtained by recalling 
the expression for $P_N(b,F_\ell)$.
Using the above inequality in the following, we obtain
\begin{eqnarray*}
J_{N-1}(b,F_\ell)
&=&\min\Big\{-\eta b, P_{N-1}(b,F_\ell), C_{N-1}(b,F_\ell)\Big\}\\
&\le&\min\Big\{-\eta b, P_{N-1}(b,F_\ell)\Big\}\\
&\le&\min\Big\{-\eta b, \eta\delta-\eta\mathbb{E}_\ell\Big[\max\{b,R_\ell\}\Big]\Big\}\\
&=& J_N(b,F_\ell).
\end{eqnarray*}
Thus we have shown the result for stage $N-1$. 

Suppose the result is true for some stage $k+1=2,3,\cdots,N-1$.
i.e., $J_{k+1}(b)\le J_{k+2}(b)$ and $J_{k+1}(b,F_\ell)\le J_{k+2}(b,F_\ell)$ (for all $F_\ell$), then,
using the induction hypothesis, the cost of continuing, $C_k(b)$, can be bounded as
\begin{eqnarray*}
C_k(b)
&=&\tau+\mathbb{E}_L\Big[J_{k+1}(b,F_{k+1})\Big]\\
&\le&\tau+\mathbb{E}_L\Big[J_{k+2}(b,F_{k+2})\Big]\\
&=&C_{k+1}(b).
\end{eqnarray*}
Thus, we have $J_k(b)\le J_{k+1}(b)$ (see (\ref{Jk_b_equn})). 
Next, consider the probing cost,
\begin{eqnarray*}
P_k(b,F_\ell)
&=&\eta\delta+\mathbb{E}_\ell\Big[J_{k}(\max\{b,R_\ell\})\Big]\\
&\overset{*}{\le}&\eta\delta+\mathbb{E}_\ell\Big[J_{k+1}(\max\{b,R_\ell\})\Big]\\
&=&P_{k+1}(b,F_\ell)
\end{eqnarray*}
where, to obtain  $*$ we have used $J_k(b)\le J_{k+1}(b)$ which we have already shown.
The cost of continuing can be similarly bounded:
\begin{eqnarray*}
C_k(b,F_\ell)
&=&\tau + \mathbb{E}_{L}\Big[\min\{J_{k+1}(b,F_\ell),J_{k+1}(b,F_{L_{k+1}})\}\Big]\\
&\overset{*}{\le}&\tau + \mathbb{E}_{L}\Big[\min\{J_{k+2}(b,F_\ell),J_{k+2}(b,F_{L_{k+2}})\}\Big]\\
&=&C_{k+1}(b,F_\ell),
\end{eqnarray*}
where $*$ is due to the induction hypothesis and 
the fact that location random variables, $L_{k+1}$
and $L_{k+2}$, are identically distributed. Finally, using the above inequalities in the 
expression of $J_k(b,F_\ell)$ \Big(recall (\ref{Jk_bF_equn}); $J_k(b,F_\ell)=\min\Big\{-\eta b, P_k(b,F_\ell), C_k(b,F_\ell)\Big\}$\Big),
we obtain $J_k(b,F_\ell)\le J_{k+1}(b,F_\ell)$, thus completing the proof.
\end{IEEEproof}

\section{Proof of Lemma~\ref{costs_bounded_lemma}}
\label{threshold_nature_lemma_appendix}
The following simple property about the $\min$-operator will be useful
while proving Lemma~\ref{costs_bounded_lemma}.
\begin{lemma}
\label{min_operator_lemma}
If $x_1,x_2,\cdots,x_j$ and $y_1,y_2,\cdots,y_j$ in $\Re$, are such that, $x_i-y_i\le x_1-y_1$ for 
all $i=1,2,\cdots,j$, then 
\begin{eqnarray}
\label{min_operator_equn}
\min\{x_1,x_2,\cdots, x_j\} -\min\{y_1,y_2,\cdots, y_j\}\le x_1-y_1
\end{eqnarray}
\end{lemma}
\begin{IEEEproof}
Suppose $\min\{y_1,y_2,\cdots, y_j\}=y_i$, for some $1\le i\le j$, then the LHS of (\ref{min_operator_equn})
can be written as,
\begin{eqnarray*}
LHS = \min\{x_1,x_2,\cdots, x_j\} - y_i
\le x_i - y_i.
\end{eqnarray*}
The proof is complete by recalling that we are given, $x_i-y_i\le x_1-y_1$.
\end{IEEEproof}

\verb11

\emph{Lemma}~\ref{costs_bounded_lemma}\emph{:}
For $k=1,2,\cdots,N-1$ (for part~(ii), $k=1,2,\cdots,N$), for any $F_\ell$, 
and for $b_2>b_1$ we have
\begin{enumerate}
\item[(i)] $C_k(b_1)-C_k(b_2)\le\eta(b_2-b_1)$,
\item[(ii)] $P_k(b_1,F_\ell)-P_k(b_2,F_\ell)\le\eta(b_2-b_1)$ 
\item[(iii)] $C_k(b_1,F_\ell)-C_k(b_2,F_\ell)\le\eta(b_2-b_1)$. 
\end{enumerate}

\begin{IEEEproof}
Since $J_N(b)$ is $-\eta b$ we already have, for stage $N$,
$J_N(b_1)-J_N(b_2)=\eta(b_2-b_1)$.
Also, for a given distribution $F_\ell$ and for $b_2>b_1$,
\begin{eqnarray*}
P_N(b_1,F_\ell)-P_N(b_2,F_\ell) 
&=&\eta\mathbb{E}_\ell\Big[\max\{b_2,R_\ell\}-\max\{b_1,R_\ell\}\Big]\\
&\overset{*}{\le}& \eta(b_2-b_1),
\end{eqnarray*}
where to obtain $*$, first consider all the three cases that are possible: 
(1) $R_\ell\le b_1< b_2$, (2) $b_1<R_\ell<b_2$, and
(3) $b_1<b_2\le R_\ell$, and then note that in all these cases, 
$\Big(\max\{b_2,R_\ell\}-\max\{b_1,R_\ell\}\Big)$, is bounded above by $b_2-b_1$. 
Now, since $J_N(b,F_\ell)=\min\Big\{-\eta b,P_N(b,F_\ell)\Big\}$,
the above inequality along with Lemma~\ref{min_operator_lemma} will yield,
$J_N(b_1,F_\ell)-J_N(b_2,F_\ell)\le\eta(b_2-b_1)$.

Suppose for some stage $k+1=1,2,\cdots,N$ we have $J_{k+1}(b_1)-J_{k+1}(b_2)\le\eta(b_2-b_1)$ and 
$J_{k+1}(b_1,F_\ell)-J_{k+1}(b_2,F_\ell)\le\eta(b_2-b_1)$ for all $b_2>b_1$, and for all $F_\ell$. 
Then we will show that
all the inequalities listed in the lemma will hold for stage $k$ as well.
First, a simple application of the induction hypothesis will yield,
\begin{eqnarray*}
C_k(b_1)-C_k(b_2) 
&=& \mathbb{E}_L\Big[J_{k+1}(b_1,F_{L_k})-J_{k+1}(b_2,F_{L_k})\Big]\\
&\le&\eta(b_2-b_1).
\end{eqnarray*}
Since $J_k(b)=\min\Big\{-\eta b, C_k(b)\Big\}$, the above inequality along with 
 Lemma~\ref{min_operator_lemma} gives, $J_k(b_1)-J_k(b_2)\le\eta(b_2-b_1)$, for any $b_2>b_1$.
Using this we can write
\begin{eqnarray}
\label{probing_diff_equn}
P_k(b_1,F_\ell)-P_k(b_2,F_\ell) 
&=& \mathbb{E}_\ell\Big[J_k(\max\{b_1,R_\ell\})-J_k(\max\{b_2,R_\ell\})\Big]\nonumber\\
&\le&\mathbb{E}_\ell\Big[\eta\Big(\max\{b_2,R_\ell\}-\max\{b_1,R_\ell\}\Big)\Big]\nonumber\\
&\le&\eta(b_2-b_1),
\end{eqnarray}
where the last inequality is again by considering all the three regions where $R_\ell$ can lie.

To show part~(iii), define $\mathcal{L}_\ell$ as the set of all distributions that are 
stochastically greater than
$\ell$, i.e., $\mathcal{L}_\ell=\Big\{F_t\in\mathcal{F}: F_t\ge_{st} F_\ell\Big\}$. 
Let $\mathcal{L}_\ell^c$ denote the set of all the remaining distributions, i.e.,
$\mathcal{L}_\ell^c=\mathcal{F}\setminus\mathcal{L}_\ell$.
From  Lemma~\ref{F_total_order_lemma},
where we have shown that $\mathcal{F}$ is totally stochastically ordered
(see Definition~\ref{total_stochastic_ordering_definition}),
it follows that $\mathcal{L}_\ell^c$ contains all distributions in $\mathcal{F}$ which are
stochastically smaller than $F_\ell$. 
Recalling the expression for $C_k(b,F_\ell)$ from (\ref{continuing_cost_bF_equn}), 
the difference in the cost of continuing can now be bounded as follows:
\begin{eqnarray}
\label{continue_diff_equn}
C_k(b_1,F_\ell)-C_k(b_2,F_\ell) 
&=& \int_\mathcal{F} \Big(\min\{J_{k+1}(b_1,F_\ell),J_{k+1}(b_1,F_t)\} \nonumber\\
&&\hspace{2cm}-\min\{J_{k+1}(b_2,F_\ell),J_{k+1}(b_2,F_t)\}\Big)dL(t) \nonumber\\
&\overset{*}{=}& \int_{\mathcal{L}_\ell} (J_{k+1}(b_1,F_t)-J_{k+1}(b_2,F_t))dL(t) \nonumber \\
&& \hspace{2cm}+\int_{\mathcal{L}_\ell^c} (J_{k+1}(b_1,F_\ell)-J_{k+1}(b_2,F_\ell))dL(t). \nonumber\\
&\overset{o}{\le}& \eta (b_2-b_1).
\end{eqnarray}
In the above derivation, $*$ is obtained by using  Lemma~\ref{cost_to_go_ordering_lemma}-(i),
and $o$ is simply by applying the induction argument. 
Since $J_k(b,F_\ell)=\min\Big\{-\eta b, P_k(b,F_\ell),C_k(b,F_\ell)\Big\}$,
using (\ref{probing_diff_equn}) and (\ref{continue_diff_equn})
along with Lemma~\ref{min_operator_lemma},
we obtain, $J_{k}(b_1,F_\ell)-J_{k}(b_2,F_\ell)\le\eta(b_2-b_1)$, thus completing 
the induction argument. 
\end{IEEEproof}

\section{Proof of Lemma~\ref{F_total_order_lemma}}
\label{F_total_order_lemma_appendix}

\emph{Lemma}~\ref{F_total_order_lemma}\emph{:}
The set of reward distributions $\mathcal{F}$ in (\ref{distribution_set}), is 
totally stochastically ordered with a minimum distribution. 

\begin{IEEEproof}
Recall the reward expression from (\ref{reward_equn}),
\begin{eqnarray*}
R_{\ell}=\frac{Z_{\ell}^a}{P_{\ell}^{(1-a)}}=\frac{Z_{\ell}^a}{(\Gamma' D_{\ell}^\xi)^{(1-a)}} G_{\ell}^{(1-a)}.
\end{eqnarray*}
The distribution, $F_\ell$,  of $R_\ell$ can be written as,
\begin{eqnarray}
\label{dist_kappa_equn}
F_\ell(r) 
&=& \mathbb{P}(R_\ell\le r) \nonumber \\
&=& \mathbb{P}\left(\frac{Z_{\ell}^a}{(\Gamma' D_{\ell}^\xi)^{(1-a)}} G_{\ell}^{(1-a)}\le r\right)\nonumber\\
&=& \mathbb{P}\left(G_{\ell}^{(1-a)}\le {\kappa}_{\ell} r\right),
\end{eqnarray}
where $\kappa_\ell=\frac{(\Gamma' D_{\ell}^\xi)^{(1-a)}}{Z_{\ell}^a}$.

Let $\ell,u$ be any two locations in $\mathcal{L}$. Since the rewards are non-negative, we have
$F_\ell(r) = F_u(r)=0$ for $r<0$. Hence, we  only need to consider the case $r\ge0$.
Now, given $\ell,u\in\mathcal{L}$, either $\kappa_\ell\le\kappa_u$ or $\kappa_\ell>\kappa_u$.
Thus we have, either $\kappa_\ell r\le\kappa_u r$ or  $\kappa_\ell r\ge\kappa_u r$, for every $r\ge0$.
Finally, since $G_\ell$ and $G_u$ are identically distributed, we have,
either $F_\ell(r)\le F_u(r)$ or $F_\ell(r)\ge F_u(r)$, for all $r$, so that $F_\ell$ and $F_u$ are 
stochastically ordered (recall Definition~\ref{stochastic_ordering_definition}).

To show that there exists a minimum distribution, first note that $\kappa_\ell$ as a function of 
$\ell\in\mathcal{L}$ is continuous. Then, since we had assumed that $\mathcal{L}$ is compact
(closed and bounded), there exists an $m\in\mathcal{L}$ where the maximum is achieved, i.e., 
$\kappa_\ell\le\kappa_m$ for all $\ell\in\mathcal{L}$. Again, since the gains $G_\ell$ 
and $G_m$ are identically distributed, from (\ref{dist_kappa_equn}) it follows
that  $F_\ell\ge_{st}F_m$ for all $\ell\in\mathcal{L}$, so that $F_m$ is the minimum distribution.
\end{IEEEproof}

\section{Proof of Lemma~\ref{equal_costs_lemma}}
\label{equal_costs_lemma_appendix}

\emph{Lemma}~\ref{equal_costs_lemma}\emph{:}
Suppose $\mathcal{S}_k\subseteq\mathcal{Q}_k^u$, for some $F_u$, and some $k=1,2,\cdots,N-1$. 
Then for every $b\in\mathcal{S}_k$ we have $J_k(b,F_u)=J_N(b,F_u)$.

\begin{IEEEproof}
Fix a $b\in\mathcal{S}_k\subseteq\mathcal{Q}_k^u$. Then, 
\begin{eqnarray*}
J_k(b,F_u)
&=&\min\Big\{-\eta b, P_k(b,F_u), C_k(b,F_u)\Big\}\\
&\overset{*}{=}&\min\Big\{-\eta b, P_k(b,F_u)\Big\}\\
&\overset{o}{=}&\min\Big\{-\eta b, \eta\delta+\mathbb{E}_u\Big[J_{k}(\max\{b,R_u\})\Big]\Big\}\nonumber\\
&\overset{\dagger}{=}&\min\Big\{-\eta b, \eta\delta-\eta \mathbb{E}_u\Big[\max\{b,R_u\}\Big]\Big\}\nonumber\\
&=& J_N(b,F_u).
\end{eqnarray*}
In the above derivation, $*$ is because, $b$ being in $\mathcal{Q}_k^u$,
at $(b,F_u)$ it is optimal to either {stop} 
or {probe} (recall (\ref{optimal_stopping_probing_equn})). $o$ is simply obtained by substituting for $P_k(b,F_u)$
from (\ref{probing_cost_bF_equn}). Further, after probing the new state, 
$\max\{b,R_u\}\ge b$, is also in $\mathcal{S}_k$
(from Theorem~\ref{threshold_nature_lemma}) so that it is optimal to {stop} after probing.
This observation yields $\dagger$. Finally, the last equality is obtained by 
recalling the expression of $J_N(b,F_u)$ from 
(\ref{cost_to_go_stageN_statebF_equn}).
\end{IEEEproof}

\section{Proof of Lemma~\ref{all_distributions_corollary}}
\label{all_distributions_corollary_appendix}
As discussed in the outline of the proof of Lemma~\ref{all_distributions_corollary},
the result immediately follows once we show \emph{Step\ 1} and \emph{Step\ 2}.
First we will formally state and prove \emph{Step\ 1}.\\

\begin{lemma}
\label{stopping_probing_set_lemma}
Suppose $F_u$ is a distribution such that for all 
$k=1,2,\cdots,N-1$,  $\mathcal{S}_k\subseteq\mathcal{Q}_k^u$.
Then for any distribution $F_\ell\ge_{st}F_u$ we have $\mathcal{S}_k\subseteq\mathcal{Q}_k^\ell$.
\end{lemma}
\begin{IEEEproof}
We will first show that $\mathcal{S}_{N-1}\subseteq\mathcal{Q}_{N-1}^\ell$.
Fix a $b\in\mathcal{S}_{N-1}$. Then $b\in\mathcal{Q}_{N-1}^u$ (because it is given that $\mathcal{S}_{N-1}\subseteq\mathcal{Q}_{N-1}^u$),
so that using the definition of the set $\mathcal{Q}_{N-1}^u$ (from (\ref{optimal_stopping_probing_equn})) we can write
\begin{eqnarray}
\label{using_defn_equn}
\min\Big\{-\eta b, P_{N-1}(b,F_u)\Big\}\le C_{N-1}(b,F_u).
\end{eqnarray}
For any generic distribution $F_s$, whenever  $b\in\mathcal{S}_{N-1}$, 
the minimum of the cost of stopping and the cost of probing
can be simplified as follows:
\begin{eqnarray}
\label{stopping_probing_cost_N-1_equn}
\min\Big\{-\eta b, P_{N-1}(b,F_s)\Big\}
&\overset{*}{=}& \min\Big\{-\eta b, \eta\delta+\mathbb{E}_s\Big[J_{N-1}(\max\{b,R_s\})\Big]\Big\}\nonumber\\
&\overset{o}{=}&\min\Big\{-\eta b, \eta\delta-\eta\mathbb{E}_s\Big[\max\{b,R_s\}\Big]\Big\}\nonumber\\
&\overset{\dagger}{=}&J_N(b,F_s).
\end{eqnarray}
In the above, $*$ is obtained by recalling the expression for the probing cost from 
(\ref{probing_cost_bF_equn}). $o$ is 
because, after probing we are still at stage $N-1$ with the new state $\max\{b,R_s\}$ 
also in $\mathcal{S}_{N-1}$ (Lemma~\ref{threshold_nature_lemma}); in 
$\mathcal{S}_{N-1}$ we know that it  is optimal to stop,
so that $J_{N-1}(\max\{b,R_s\})=-\eta \max\{b,R_s\}$. Finally, to obtain $\dagger$,
recall the expression for $J_N(b,F_s)$ from (\ref{cost_to_go_stageN_statebF_equn}).

Now using (\ref{stopping_probing_cost_N-1_equn}) in (\ref{using_defn_equn})
we see that, the hypothesis $b\in\mathcal{S}_{N-1}$ implies,  $J_N(b,F_u)\le C_{N-1}(b,F_u)$.
Also, from Lemma~\ref{cost_to_go_ordering_lemma}-(i)  we have,
$J_N(b,F_\ell)\le J_N(b,F_u)$ for any $F_\ell\ge_{st}F_u$. Combining these we can write
\begin{eqnarray}
\label{xyz_equn}
J_N(b,F_\ell)\le J_N(b,F_u)\le C_{N-1}(b,F_u).
\end{eqnarray}
To conclude that $b\in\mathcal{Q}_{N-1}^\ell$, we need to show
\begin{eqnarray*}
\min\Big\{-\eta b, P_{N-1}(b,F_\ell)\Big\}\le C_{N-1}(b,F_\ell),
\end{eqnarray*}
or, alternatively, recalling (\ref{stopping_probing_cost_N-1_equn}),
it is sufficient to show, 
\begin{eqnarray}
\label{to_show_equn}
J_N(b,F_\ell)\le C_{N-1}(b,F_\ell).
\end{eqnarray}

Now for any generic distribution $F_s\in\mathcal{F}$ define 
$\mathcal{L}_s=\Big\{t\in\mathcal{L}: F_t\ge_{st} F_s\Big\}$ 
i.e., $\mathcal{L}_s$ is the set of all distributions in $\mathcal{F}$ that 
are stochastically greater than $F_s$. 
Let $\mathcal{L}_\ell^c$ denote the set of all the remaining distributions, i.e.,
$\mathcal{L}_\ell^c=\mathcal{F}\setminus\mathcal{L}_\ell$.
Since $\mathcal{F}$ is totally stochastically ordered 
(Lemma~\ref{F_total_order_lemma}), $\mathcal{L}_s^c$
contains all distributions in $\mathcal{F}$ that are stochastically smaller than $F_s$. Further,
for $F_\ell\ge_{st}F_u$ we have $\mathcal{L}_\ell\subseteq\mathcal{L}_u$.
Then, recalling the expression for $C_{N-1}(b,F_u)$ from (\ref{continuing_cost_bF_equn}) we can write
\begin{eqnarray}
C_{N-1}(b,F_u)
&=&\tau + \mathbb{E}_{L}\Big[\min\{J_{N}(b,F_u),J_{N}(b,F_{L_{N}})\}\Big]\nonumber\\
&\overset{*}{=}&\tau + \int_{\mathcal{L}_u}J_{N}(b,F_t)\ dL(t) + \int_{\mathcal{L}_u^c}J_{N}(b,F_u)\ dL(t)\nonumber\\
&\overset{o}{=}&\tau + \int_{\mathcal{L}_\ell}J_{N}(b,F_t)\ dL(t) + \int_{\mathcal{L}_u\setminus\mathcal{L}_\ell}J_{N}(b,F_t)\ dL(t) + \int_{\mathcal{L}_u^c}J_{N}(b,F_u)\ dL(t),\nonumber
\end{eqnarray}
where, $*$ is obtained by using Lemma~\ref{cost_to_go_ordering_lemma}-(i) 
and the definition of $\mathcal{L}_u$, and to obtain $o$ we have split the integral over $\mathcal{L}_u$
(first integral in $*$) into two integrals $-$ one over $\mathcal{L}_\ell$ and the other over
$\mathcal{L}_u\setminus\mathcal{L}_\ell$.
Now, for any $F_t\in\mathcal{L}_u\setminus\mathcal{L}_\ell$ we know that $F_t\ge_{st}F_u$ so that 
$J_{N}(b,F_t)\le J_{N}(b,F_u)$ (again from Lemma~\ref{cost_to_go_ordering_lemma}-(i)). 
Thus, in the above expression,
replacing $J_N(b,F_t)$ by $J_N(b,F_u)$ in the middle integral, and then combining it with the 
last integral, we obtain
\begin{eqnarray}
\label{xy_py_equn}
C_{N-1}(b,F_u)&\le& \tau + \int_{\mathcal{L}_\ell}J_{N}(b,F_t)\ dL(t) +  
\left(\int_{\mathcal{L}_\ell^c}dL(t)\right)\ J_{N}(b,F_u)
\end{eqnarray}
From (\ref{xyz_equn}) and (\ref{xy_py_equn}) we see that we have an inequality of the following form
\begin{eqnarray}
\label{three_order_equn}
J_N(b,F_\ell)\le J_N(b,F_u)\le c+pJ_N(b,F_u),
\end{eqnarray}
where $c=\tau + \int_{\mathcal{L}_\ell}J_{N}(b,F_t)\ dL(t)$ and $p=\int_{\mathcal{L}_\ell^c}dL(t)$. Since
$p\in[0,1]$ we can write 
\begin{eqnarray*}
J_N(b,F_\ell)(1-p)\le J_N(b,F_u)(1-p),
\end{eqnarray*}
rearranging which we obtain,
\begin{eqnarray*}
J_N(b,F_\ell)&\le& pJ_N(b,F_\ell)+J_N(b,F_u)-pJ_N(b,F_u)\nonumber\\
&\overset{*}{\le}& pJ_N(b,F_\ell)+c+pJ_N(b,F_u)-pJ_N(b,F_u)\nonumber\\
&=&c+pJ_N(b,F_\ell)
\end{eqnarray*}
where, to obtain $*$ we have used (\ref{three_order_equn}).
Finally, note that
\begin{eqnarray*}
c+pJ_N(b,F_\ell) 
&=& \tau + \int_{\mathcal{L}_\ell}J_{N}(b,F_t)\ dL(t) + \left(\int_{\mathcal{L}_\ell^c}dL(t)\right) J_N(b,F_\ell)\\
&=& \tau + \mathbb{E}_{L}\Big[\min\{J_{N}(b,F_\ell),J_{N}(b,F_{L_{N}})\}\Big]\\
&=& C_{N-1}(b,F_\ell). 
\end{eqnarray*}
Thus, as desired we have shown $J_N(b,F_\ell)\le C_{N-1}(b,F_\ell)$ 
(recall the discussion leading to (\ref{to_show_equn})).

Suppose that for some $k+1=2,3,\cdots,N-1$ we have $\mathcal{S}_{k+1}\subseteq\mathcal{Q}_{k+1}^\ell$. We will have to show
that the same holds for stage $k$. Fix any $b\in\mathcal{S}_k$, then
for any generic distribution $F_s$, exactly as in (\ref{stopping_probing_cost_N-1_equn}) we have
\begin{eqnarray}
\min\Big\{-\eta b, P_k(b,F_s)\Big\} 
&=&\min\Big\{-\eta b, \eta\delta+\mathbb{E}_s\Big[J_{k}(\max\{b,R_s\})\Big]\Big\} \nonumber\\
&=& \min\Big\{-\eta b, \eta\delta-\eta\mathbb{E}_s\Big[\max\{b,R_s\}\Big]\Big\}\nonumber\\
&=& J_{N}(b,F_s).
\end{eqnarray}
Thus the hypothesis $\mathcal{S}_k\subseteq\mathcal{Q}_k^u$ implies $J_N(b,F_u)\le C_k(b,F_u)$,
and to show 
$\mathcal{S}_k\subseteq\mathcal{Q}_k^\ell$ it is sufficient to obtain
$J_N(b,F_\ell)\le C_k(b,F_\ell)$.
Proceeding as before (recall how (\ref{xy_py_equn}) was obtained) we can write
\begin{eqnarray*}
C_{k}(b,F_u)
&\le& \tau + \int_{\mathcal{L}_\ell}J_{k+1}(b,F_t)\ dL(t) +  
\left(\int_{\mathcal{L}_\ell^c}dL(t)\right)\ J_{k+1}(b,F_u).
\end{eqnarray*}
Now using Lemma~\ref{equal_costs_lemma}, we conclude
\begin{eqnarray*}
C_{k}(b,F_u)&\le&\tau + \int_{\mathcal{L}_\ell}J_{k+1}(b,F_t)\ dL(t) +  \int_{\mathcal{L}_\ell^c}dL(t)\ J_{N}(b,F_u).
\end{eqnarray*}
Note that the conditions required to apply Lemma~\ref{equal_costs_lemma} 
hold i.e., $b\in\mathcal{S}_{k+1}$ (since $\mathcal{S}_k\subseteq\mathcal{S}_{k+1}$ 
from Lemma~\ref{sets_trivial_ordering_corollary}-(iii)) and $\mathcal{S}_{k+1}\subseteq\mathcal{Q}_{k+1}^u$ (this is given).

Thus, again we have an inequality of the form $J_{N}(b,F_\ell)\le J_{N}(b,F_u)\le c'+p J_{N}(b,F_u)$ 
(where $c'=\tau + \int_{\mathcal{L}_\ell}J_{k+1}(b,F_t)\ dL(t)$). As before we can show that
$J_N(b,F_\ell)\le c'+p J_{N}(b,F_\ell)$.
Finally the proof is complete by showing that  $c'+p J_{N}(b,F_\ell)=C_k(b,F_\ell)$ as follows:
\begin{eqnarray}
C_k(b,F_\ell)&=&\tau + \int_{\mathcal{L}_\ell}J_{k+1}(b,F_t)\ dL(t) + \int_{\mathcal{L}_\ell^c}J_{k+1}(b,F_\ell)\ dL(t)\nonumber\\
&=&c'+p J_{N}(b,F_\ell),
\end{eqnarray}
where to replace $J_{k+1}(b,F_\ell)$ by $J_N(b,F_\ell)$ we have to again apply Lemma~\ref{equal_costs_lemma}.
However this time $\mathcal{S}_{k+1}\subseteq\mathcal{Q}_{k+1}^\ell$, is by the induction hypothesis.
\end{IEEEproof}

\verb11

We still require a distribution $F_u$ satisfying $\mathcal{S}_k\subseteq\mathcal{Q}_k^u$,
for every $k$. The minimum distribution $F_m$ turns out to be useful in this context.
The following lemma thus constitutes \emph{Step\ 2} of the proof of Lemma~\ref{all_distributions_corollary}.

\begin{lemma}
\label{dominated_set_lemma}
For every $k=1,2,\cdots,N-1$, the stobing set $\mathcal{Q}_k^m$ corresponding to the minimum
distribution $F_m$ satisfies, $\mathcal{S}_k\subseteq\mathcal{Q}_k^m$.
\end{lemma}
\begin{IEEEproof}
First note that the existence of a minimum distribution $F_m$ follows from Lemma~\ref{F_total_order_lemma}.
Now,  $F_m$ being minimum we have $F_\ell\ge_{st}F_m$ for all $F_\ell$. Then, using
Lemma~\ref{cost_to_go_ordering_lemma}-(i) we can write
\begin{eqnarray*}
J_{k+1}(b,F_{L_{k+1}})\le J_{k+1}(b,F_m).
\end{eqnarray*}
Using the above expression in (\ref{continuing_cost_bF_equn}) and then recalling 
(\ref{continuing_cost_b_equn}),  we obtain $C_k(b,F_m)=C_k(b)$. Finally, the
result follows from the definition of 
the sets $\mathcal{Q}_k^m$ and $\mathcal{S}_k$.
\end{IEEEproof}

\section{Proof of Theorem~\ref{stopping_sets_l_equal_theorem}}
\label{stopping_sets_l_equal_theorem_appendix}

\emph{Theorem}~\ref{stopping_sets_l_equal_theorem}\emph{:}
For $k=1,2,\cdots,N-1$ and for any $F_\ell$,
$\mathcal{S}_k^\ell=\mathcal{S}_{k+1}^\ell$.

\begin{IEEEproof}
Recalling the definition of the set $\mathcal{S}_k^\ell$ (from (\ref{optimal_stopping_l_equn})),
for any $b\in\mathcal{S}_{k+1}^\ell$ we have (if $k+1=N$, note that the following expression will not contain
the continuing cost),
\begin{eqnarray*}
-\eta b\le \min\Big\{P_{k+1}(b,F_\ell),C_{k+1}(b,F_\ell)\Big\}.
\end{eqnarray*}
Suppose, as in Theorem~\ref{stopping_sets_equal_theorem}, we can 
show that for any $b\in\mathcal{S}_{k+1}^\ell$, the various costs at stages $k$ and $k+1$ are same, i.e.,
$P_k(b,F_\ell)=P_{k+1}(b,F_\ell)$ and $C_k(b,F_\ell)=C_{k+1}(b,F_\ell)$, then the above inequality would imply,
$\mathcal{S}_k^\ell\supseteq\mathcal{S}_{k+1}^\ell$.
The proof is complete by recalling that we already have 
$\mathcal{S}_k^\ell \subseteq \mathcal{S}_{k+1}^\ell$ (from Lemma~\ref{sets_trivial_ordering_corollary}-(iii)).

Fix a $b\in\mathcal{S}_{k+1}^\ell$. To show that
$P_k(b,F_\ell)=P_{k+1}(b,F_\ell)$, first 
using Lemma~\ref{sets_trivial_ordering_corollary}-(i)
and Theorem~\ref{stopping_sets_equal_theorem}, note that
$\mathcal{S}_{k+1}^\ell\subseteq\mathcal{S}_{k+1}=\mathcal{S}_k$.
Since $b\in\mathcal{S}_{k+1}$ the cost of probing is
\begin{eqnarray*}
P_{k+1}(b,F_\ell) 
&=& \eta\delta + \mathbb{E}_\ell\Big[J_{k+1}(\max\{b,R_\ell\})\Big]\\
&=& \eta\delta -\eta\mathbb{E}_\ell\Big[\max\{b,R_\ell\}\Big]
\end{eqnarray*}
where, to obtain the second equality, note that $\max\{b,R_\ell\}\in\mathcal{S}_k$ (from Theorem~\ref{threshold_nature_lemma})
and hence at $\max\{b,R_\ell\}$ it is optimal to stop, so that $J_{k+1}(\max\{b,R_\ell\})=-\eta \max\{b,R_\ell\}$.
Similarly, since $b$ is also in $\mathcal{S}_k$ the cost of probing at stage $k$, $P_k(b,F_\ell)$,
is again $\eta\delta -\eta\mathbb{E}_\ell\Big[\max\{b,R_\ell\}\Big]$.
Finally, following the same procedure used to show $C_k(b)=C_{k+1}(b)$ in Theorem~\ref{stopping_sets_equal_theorem},
we can obtain $C_k(b,F_\ell)=C_{k+1}(b,F_\ell)$, thus completing the proof.
\end{IEEEproof}

\end{document}